\documentclass{article}

\usepackage[preprint]{neurips_2026}
\usepackage[utf8]{inputenc} 
\usepackage[T1]{fontenc}    
\usepackage{hyperref}       
\usepackage{url}            
\usepackage{booktabs}       
\usepackage{amsfonts}       
\usepackage{nicefrac}       
\usepackage{microtype}      
\usepackage{xcolor}         
\usepackage{stmaryrd,} 
\usepackage{amsthm, dsfont, xparse, stmaryrd}
\usepackage{amssymb}
\usepackage{mathtools}
\usepackage{enumitem}
\usepackage[header,title,page]{appendix}
\usepackage{etoc}
\usepackage{thmtools,thm-restate}
\usepackage{pgfplots}
\pgfplotsset{compat=1.18} 
\usepgfplotslibrary{fillbetween}
\usepgfplotslibrary{groupplots}
\usepackage{lmodern}
\usepackage{graphicx}
\usepackage{algorithm}
\usepackage{algpseudocode}
\usepackage{booktabs}
\usepackage{float}

\title{$\alpha$-fair Heterogeneous Agent Reinforcement Learning}

\author{
    Yao-hua Franck Xu \\
    Orange Innov \\
    Lannion, France \\
    \texttt{franck.xu@orange.com} \\ 
    \And
    Tayeb Lemlouma \\
    IRISA \\
    Université de Rennes \\
    Lannion, France \\
    tayeb.lemlouma@irisa.fr \\
    \And
    Jean-Marie Bonnin \\
    IRISA \\
    IMT Atlantique \\
    Rennes, France \\
    jean-marie.bonnin@irisa.fr\\
    \And
    Arnaud Braud \\
    Orange Innov \\
    Lannion, France \\
    arnaud.braud@orange.com
}

\newtheorem{theorem}{Theorem}[section]
\newtheorem{corollary}{Corollary}[theorem]
\newtheorem{lemma}[theorem]{Lemma}

\theoremstyle{definition}
\newtheorem{definition}{Definition}[section]

\theoremstyle{remark}

\theoremstyle{definition}
\newtheorem{assumption}{Assumption}

\newcommand{\E}{\mathbb{E}}
\newcommand{\diffj}{J(\vec{\pi}')-J(\vec{\pi})}
\newcommand{\dkl}[2]{D_{\text{KL}}^{\max}\left(#1\,,\,#2\right)}
\newcommand{\ua}{U_{\alpha}}
\NewDocumentCommand{\sumi}{O{j}}{\sum_{#1=1}^n}
\newcommand{\sumt}{\sum_{t=0}^\infty}
\NewDocumentCommand{\vf}{O{} O{j}}{V_{#2}^{\vec{\pi}#1}}
\NewDocumentCommand{\calvf}{O{} O{j}}{\mathcal{V}_{#2}^{\vec{\pi}#1}}
\newcommand{\void}{\varnothing}

\begin{document}

\maketitle

\begin{abstract}
  Cooperation in multi-agent systems is typically optimized through utilitarian objectives that maximize overall efficiency but fail to account for reward distribution, often resulting in inequitable "leader-follower" dynamics. While fairness-based approaches encourage pro-social behaviors where every agent benefits from cooperation, many current algorithms—including those utilizing reward shaping—break the stationarity of Markov Games or lack rigorous theoretical guarantees. This creates a critical gap between fair objective methods and theoretically safe learning frameworks. We propose a novel framework that bridges $\alpha$-fairness with Heterogeneous-Agent Trust Region Learning (HATRL), ensuring monotonic improvement and convergence toward Nash Equilibria. Our approach leverages a fair advantage function that dynamically weights agent utilities based on their expected returns, allowing the global objective to transition from purely utilitarian efficiency to $\alpha$-fairness welfare based on the parameter $\alpha$. We introduce two practical algorithms, $\alpha$-fair HATRPO and $\alpha$-fair HAPPO, and demonstrate through experiments in sequential social dilemmas like CleanUp and CommonHarvest that they perform better than HATRL's algorithms from a utilitarian point of view while achieving socially higher outcomes.
\end{abstract}


\section{Introduction}
Cooperation is a fundamental pillar of multi-agent systems (MAS), ranging from biological systems (\cite{biology_cooperation}) to complex human societies(\cite{cooperation_society}). Within Multi-Agent Reinforcement Learning (MARL), cooperation has traditionally been viewed through a utilitarian lens (\cite{selfishness, utilitarian}), where the primary goal is to maximize the sum of all agents' utilities. However, this focus on raw efficiency often ignores how rewards are distributed, potentially creating "leader-follower" dynamics (\cite{moi}) where a subset of agents is marginalized to optimize the collective total. For cooperation to be sustainable, fairness is not just a moral preference but a functional necessity, ensuring that every participant gains from the interaction. This need is further amplified by the recent rise of collective decision-making driven by Large Language Models (LLMs, \cite{tran_multi-agent_2025}).

Recent efforts to integrate fairness into MARL have followed three primary paths: reward shaping, fairness constraints and social welfare objectives. Reward shaping methods, such as inequity aversion (IA, \cite{inequity_aversion}), penalize differences in temporal rewards to solve social dilemmas. However, these approaches often rely on history-dependent rewards (FEN \cite{FEN}, Fair\&LocalIA\cite{stefano}, IA\cite{inequity_aversion}), which break the stationary property of Markov Games. Furthermore, they frequently utilize Proximal Policy Optimization (PPO--\cite{trl}) based algorithms like Independent PPO (IPPO--\cite{ippo}) or Multi-Agent PPO (MAPPO--\cite{mappo}), which, while popular, lack the rigorous theoretical convergence guarantees found in single-agent Trust Region Learning (TRL -- \cite{trl}) methods. FCGrad (\cite{gradient_conflict}) adopts an interesting approach by resolving gradient conflicts of individual objective and fair objective but it lacks equilibrium analysis and interpretability. Alternatively,  fairness can be formulated as a formal constraint within the optimization problem (AdaFair-MARL\cite{marl_hospital}, DeCOM\cite{constraint_1}, EcoFair-CH-MARL\cite{constraint_2}), often by bounding specific indices like the Gini coefficient. While these constraint-based methods offer high interpretability, they are frequently domain-specific and struggle to generalize across a diverse range of multi-agent environments. Last but not least, social welfare objectives like the Generalized Gini function (\cite{gen_gini_func}), Proportional Fairness (\cite{prop_fairness}) or $\alpha$-fairness (\cite{alpha_fairness_1, alpha_fairness}) offer a more unified goal (Fair MARL\cite{moi}, SOTO\cite{soto}, AT-FAPPO \cite{at_fappo}, or \cite{ju_achieving_2023}), yet their practical implementations still largely rely on extrapolated PPO objectives that lack theoretical safeguards, and existing theoretical guarantees (e.g., in SOTO \cite{soto}) often rely on idealized assumptions (e.g., concavity of the objective function).

In this work, we bridge this gap by introducing a theoretically grounded framework for fair multi-agent learning built upon Heterogeneous-Agent Trust Region Learning (HATRL--\cite{hatrl, haml}). By combining the $\alpha$-fair objective with the monotonic improvement guarantees of HATRL, we provide a mathematically safe path for agents to learn pro-social behaviors. Our contributions are three-fold.
We first introduce the Markov Game (MG) framework, then the fair objective used through this work. We also introduce the HATRL framework in the fully-cooperative setting.
Secondly, we build a HATRL framework for our fair objective called $\alpha$-fair HATRL and study the theoretical properties of the learning framework. 
Thirdly, we develop two algorithms called $\alpha$-fair HATRPO and $\alpha$-fair HAPPO that extend $\alpha$-fair HATRL to practical settings.
Finally, we test our algorithms on two sequential social dilemmas (SSD--\cite{inequity_aversion}) environments called CleanUp and CommonHarvest and compare them to other methods relying on the same fair objective function. We demonstrate that $\alpha$-fair HATRPO and $\alpha$-fair HAPPO both reach higher level of performance compared to HATRL's algorithm while reaching fairer outcomes.


\section{Preliminaries}
\label{preli}
In this section, we first introduce the problem formulation and notations for MARL in the Markov Game settings, and then the core ideas of Heterogeneous-Agent Trust Region Learning (HATRL).

\subsection{Notations}
We denote the set of real numbers by $\mathbb{R}$ and positive real numbers by $\mathbb{R}^+$ and $d$-dimensional real-numbered vector spaces as $\mathbb{R}^d$. Interval of integers between $a$ and $b$ are denoted by $\llbracket a~;~ b \rrbracket
$. By $\mathcal{P}(A)$, we denote the power set of a set $A$. We write $|A|$ for the cardinality of a finite set $A$. We use the vector notation $\vec{.}$ to refer to any element in a cartesian product, i.e for any sets $A_1,\,\dots,\,A_n$, $\vec{a}=(a_1,\,\dots,\, a_n)\in\bigtimes_{i=1}^n A_i$.

\subsection{General-Sum Markov Games and Fair Objective}
For modeling sequential decision-making in multi-agent settings, we use the standard Markov Game framework (also known as a Stochastic Game--\cite{marl_book}). In this work, we consider a general-sum Markov game, meaning there is no condition on the reward functions.

A Markov game $\mathcal{M}$ is defined by the following tuple: $(N, \tilde{\mathcal{S}}, \allowbreak  \{\mathcal{A}_i\}_{i \in N},   P, \rho_0, \{r_i\}_{i \in N}, \gamma)$, with $N$ the finite set of $n$ players, $\tilde{\mathcal{S}}$ the space of all possible states of the environment, $\{\mathcal{A}_i\}_{i\in N}$ the individual action set for each player, $P(\tilde{s}',\vec{a}, \tilde{s}) = \mathbb{P}(s'\,|\, s, \vec{a})$ the transition probability kernel, $\rho_0\in\Delta(\tilde{\mathcal{S}})$ the initial state distribution, $\{r_i\}_{i\in N}$ the individual reward for each player and $\gamma \in [0, 1)$ the discount factor. 
At time step $t\in\mathbb N$, the environment is at state $s_t\in\mathcal S$ and each agent has to take an action $a_i\in\mathcal A_i$ according to their policy $\pi^i(\cdot|s_t)$. Given the environment state $s_t$ and the joint action $\vec a\in\mathcal A$ drawn by their joint policy $\vec\pi(\cdot|s_t)=\prod_{i=1}^n\pi^i(\cdot|s_t)$, the environment moves to the next state $s_{t+1}\in\mathcal S$ sampled by the transition probability kernel $P(\tilde{s}',\vec{a}, \tilde{s}) = \mathbb{P}(\tilde s'\,|\, \tilde s, \vec{a})$ and each agent receive its reward $ r_i(s_t,\vec a _t)\in\mathbb R^+$. The sequence $\tau=(s_0,\vec a_0, \vec r_0, s_1, \vec a_1, \vec r _1, \dots)$ is called a trajectory. All together, the discount factor, the joint policy and the transition kernel induce an improper state visitation distribution $d_{\vec\pi}$.

Let us define the contextualized state $s=(\tilde{s}, \tilde{c})\in \tilde{\mathcal{S}}\times\tilde{\mathcal{S}}$, where $\tilde{c}$ is the context and stored the initial state through trajectories. The latter will be useful to get rid of the initial-state dependency in the next equations. For any state space, we define the state value function and the state-action value function linked to agent $i$ respectively: $V_i^{\vec\pi}(s)=\E_{\tau}\left[\sumt\gamma^t r_t^i(s_t, \vec a_t)~|~s_0=s\right]$ and $Q_i^{\vec\pi}(s, \vec a)=\E_{\tau}\left[\sumt\gamma^t r_t^i(s_t, \vec a_t)~|~s_0=s, \, \vec a_0=\vec a\right]$. The advantage function associated to agent $i$ is $A_i^{\vec\pi}(s, \vec a)=Q_i^{\vec\pi}(s, \vec a) -V_i^{\vec\pi}(s)$.

In this work, we assume the state and action spaces to be finite and the reward functions to be positive and bounded by $R_{\max}$.

\subsection{Heterogeneous Agent Trust Region Learning}
HATRL \cite{hatrl} was built on top of the single-agent Trust Region Learning (TRL) framework, where one agent tries to maximize a single reward signal $r$. TRL relies on a first-order approximation of the expected return $J_{\text{coop}}(\pi)=\E_{\tau\sim\pi}\left[\sumt \gamma ^t r_t\right]$. Formally, for any policies $\pi$ and $\pi'$, TRL gives a lower bound of $J_{\text{coop}}(\pi')-J_{\text{coop}}(\pi)$ using the surrogate function $L_\pi(\tilde \pi)=J_{\text{coop}}(\pi)+\E_{s\sim d_\pi, a\sim\tilde \pi}[A^\pi(s, a)]$ and $\dkl{\pi}{\pi'}=\max_sD_{\text{KL}}(\pi(\cdot|s) \parallel \tilde \pi(\cdot|s))$: 
\begin{equation}
    J_{\text{coop}}(\pi')\geq L_\pi(\tilde\pi)-C\dkl{\pi}{\pi'},
\end{equation}
where $C=\frac{4\gamma\max_{s,a}|A^\pi(s,a)|}{(1-\gamma)^2}$.

HATRL extends the TRL framework by considering multiple heterogeneous agents trying to maximize a single reward function. HATRL adapts TRL by considering joint policies instead of single policy and leveraging sequential update to coordinate agent's learning. 

Let $\sigma$ be a permutation of $\llbracket 1~;~n\rrbracket$. Let $I_m=\{i_1, \dots,i_m\}$ denote an ordered subset of the shuffled interval, where $i_k=\sigma(k)$. Let $-I_m$ refer to its complement. HATRL defines the multi-agent state-action value function as 
\begin{equation}
    Q^{\vec \pi,\, I_m}(s, \, \vec{a}^{I_m})=\E_{\vec a^{-I_m}\sim \vec \pi_{-I_m}}\left[Q^{\vec \pi}(s,\, \vec{a}^{I_m}, \vec{a}^{-I_m})\right],
\end{equation}
and for disjoint sets $J_k$ and $I_m$, the multi-agent advantage function is defined as
\begin{equation}
    A^{\vec \pi,\, I_m}(s, \vec a^{J_k}, \vec{a}^{I_m}) = Q^{\vec \pi,\, J_k+ I_m}(s, \, \vec{a}^{J_k+ I_m}) - Q^{\vec \pi,\, J_k}(s, \, \vec{a}^{J_k}).
\end{equation}

Formally, HATRL can be described as followed. Let $\vec\pi$ be the current joint policy, $\vec \pi '$ be any other policy and $I_n$ a shuffled order. Let us define $L_{\vec\pi}^{I_m}(\vec\pi'^{I_{m-1}}, \pi'^{i_m})= \E_{s\sim d_{\vec\pi}, \vec a ^{I_{m-1}}, a^{i_m}\sim\pi'^{i_m}}[A^{\vec\pi, i_m}(s, \vec a^{I_{m-1}}, a^{i_m})]$.
HATRL states that TRL's surrogate function applied on joint policy can be decomposed into $n$ ordered summands, the $m$-th summand depending only on the $m$ first policies:
\begin{equation}
    J_{\text{coop}}(\vec \pi') \geq   J_{\text{coop}}(\vec\pi) +\sumi[m]L_{\vec \pi}^{I_m}(\vec{\pi}'^{I_{m-1}}, \pi'^{i_m}) - C\cdot \dkl{\pi^{i_m}}{\pi'^{i_m}}.
\end{equation}
The $m$-th summand can be interpreted as a first order approximation of $J(\vec\pi')-J(\vec\pi)$ with respect to the $m$-th policy for the given order $I_n$. Hence, instead of maximizing directly $J(\vec\pi)$, HATRL tries to maximize each summand of the surrogate function sequentially:
\begin{equation*}
    \pi_{k+1}^{i_m}=\underset{\pi^{i_m}}{\arg\max} \;L_{\vec \pi_k}^{I_m}(\vec{\pi}_{k+1}^{I_{m-1}}, \pi^{i_m}) - C\cdot \dkl{\pi_k^{i_m}}{\pi^{i_m}}.
\end{equation*}
To apply this procedure in practical settings with parametrised policies $\pi_{\theta}^i$ (we usually confound $\pi_\theta^i$ with its parameters $\theta^i$ for simplicity), \cite{hatrl} developed Heterogeneous-Agent TRPO (HATRPO) by approximating the KL-penalty with a KL-constraint. The authors also proposed Heterogeneous-Agent PPO (HAPPO), a first order approximation of HATRPO's objective (see \cite{hatrl} Section 4  for more details, or \cite{haml} for broader understanding).

Notice that HATRL can be easily adapted to GSMG by considering the sum of all the reward. Because the expectation is linear, we can define the global value function by summing agent's value functions (e.g. $A^{\vec\pi}(s,\vec a)=\sumi A_i^{\vec\pi}(s, \vec a)$). In that case, the global objective is exactly the utilitarian objective. 

\subsection{$\alpha$-fairness function}
In networking, $\alpha$-fairness (alpha-fairness) is an utility function used to determine how resources—typically bandwidth—should be distributed among users. Formally, it is expressed as:
\begin{equation*}
    U_\alpha(x)=
    \begin{cases}
    \frac{x^{1-\alpha}}{1-\alpha} & \text{for } \alpha \neq 1, \\
    \ln(x) & \text{for } \alpha = 1.
\end{cases}
\end{equation*}

The parameter $\alpha$ controls the level of fairness. When $\alpha\rightarrow 0$, the efficiency is preferred over how resources are distributed among users. On the other hand, when $\alpha$ tends to $+\infty$, the network just cares about the user with the lower resources. For $\alpha=1$, we obtain the proportional fairness objective (\cite{prop_fairness}), known to well-balance overall performance and fairness.


\section{Heterogeneous Trust Region Learning with Fair Objective}
In this section, the goal is to adapt HATRL for our global fair objective and build the theoretical foundation to learn fair policies. We first introduce a surrogate function for our fair global objective, then we will adapt the sequential update scheme to our new surrogate function and propose a new procedure. Finally, we will study the convergence of our method. 

\subsection{Fair HATRL}
In this work,  we consider a \textit{fully-cooperative} setting where all the agents aim to maximize the same \textit{fair global objective}:
\begin{equation*}
    J(\vec{\pi})=\mathbb{E}_{s_0\sim\rho_0}\left[\sum_{i=1}^n U_\alpha \left(\nu+V_i^{\vec{\pi}}(s_0)\right)\right],
\end{equation*}
where $\nu>0$ is a constant to ensure positiveness of the inner expression and $U_\alpha$ is the $\alpha$-fair function.

This objective can model different type of objective functions. When $\alpha\rightarrow 0$, the global objective tends to the \textit{utilitarian objective} $J_{\text{util}}(\vec\pi)=\mathbb{E}_{s_0\sim\rho_0}\left[\sum_{i=1}^n V_i^{\vec{\pi}}(s_0)\right]$. On the other hand, when $\alpha$ tends to $+\infty$, the global objective  is equivalent to \textit{the Rawlsian welfare function} $J_{\text{Rawls}}(\vec \pi)=\E_{s_0\sim\rho_0}\left[\min_{i}V_i^{\vec \pi}(s_0)\right]$. For $\alpha=1$, the global objective becomes the proportional fairness objective $J_{\text{Prop}}(\vec\pi)=\E_{s_0\sim\rho_0}\left[\sumi[i] \log V_i^{\vec \pi}(s_0)\right]$.

To adapt HATRL to our global fair objective,  we first introduce some notations. 
\begin{definition}
    Let $\vec\pi$ be the joint policies. For any contextualized state $s=(\tilde s, \tilde c)$ and joint actions $\vec a\in\mathcal A$, let us define 
    $\mathcal{V}_j^{\vec\pi}(\tilde s, \tilde c)=\left(\nu+V_j^{\vec{\pi}}(\tilde c, \tilde c)\right)^\alpha$ and the fair advantage function 
    \begin{equation}
    \label{equ:fair_advantage_function}
        A_{F}^{\vec\pi}(s,\vec a)=\sumi[j]\frac{A_j^{\vec\pi}(s,\vec a)}{\mathcal V_j^{\vec\pi}(s)}.
    \end{equation}
\end{definition}

We start by introducing a fundamental lemma which presents a bound for the difference $J(\vec\pi')-J(\vec\pi)$ with a surrogate function that is similar the one introduced in TRL. 

\begin{restatable}{lemma}{fairsurrogate} \label{lem:fair_surrogate}
    Let $\vec\pi$ be a policy and $\vec\pi'$ be an other policy. Let us define the surrogate function $\mathcal L_{\vec\pi}(\vec\pi')=\mathbb{E}_{s\sim d_{\vec\pi}, \; \vec{a}\sim\vec{\pi}'(\cdot|s)}\left[A_{F}^{\vec\pi}(s,\vec a) \right].$
    
    For the global objective defined previously, the difference $J(\vec{\pi}')-J(\vec{\pi})$ is bounded and the following inequality 
        \begin{equation*}
        \mathcal L_{\vec\pi}(\vec\pi') \geq J(\vec{\pi}')-J(\vec{\pi}) \geq \mathcal L_{\vec\pi}(\vec\pi')-C\cdot D_{\text{KL}}^{\max}(\vec\pi\;||\; \vec\pi')
    \end{equation*}
    holds, where $C=\frac{4n\alpha \nu^{-1-\alpha}}{(1-\gamma)^2}A_{\max}^2 + \frac{4nA_{\max}\gamma}{\nu^\alpha(1-\gamma)^2 }$ and $A_{\max}=\max_{j,s,\vec a}|A_j^{\vec \pi}(s,\vec a)|$
\end{restatable}

See Appendix \ref{sec:proof_lemma_3.1} for the proof. The surrogate function is a weighted sum of agent's advantage function whose weights are dynamically changing and depends on the expected return of each agent for a given initial state. The more an agent is efficient, the greater is its expected return and the less important is its advantage function in the sum -- consequently giving more space for other to thrive. Notice that the upper bound is unnecessary for building Fair HATRL but it will be very helpful in the theoretical analysis of HATRL. Now, the lower bound settled, we can introduce the value functions in the HATRL's style.
\begin{definition}
    Let $\sigma$ be a permutation of $\llbracket 1~;~n\rrbracket$. Let $I_m=\{i_1, \dots,i_m\}$ denote an ordered subset of the shuffled interval, where $i_k=\sigma(k)$. Let $-I_m$ refer to its complement. We define the agent $j$'s state-action value function as 
    \begin{equation}
        Q_j^{\vec \pi,\, I_m}(s, \, \vec{a}^{I_m})=\E_{\vec a^{-I_m}\sim \vec \pi^{-I_m}}\left[Q_j^{\vec \pi}(s,\, \vec{a}^{I_m}, \vec{a}^{-I_m})\right],
    \end{equation}
    and for disjoint sets $J_k$ and $I_m$, the agent $j$'s advantage function is defined as
    \begin{equation}
        A_j^{\vec \pi,\, I_m}(s, \vec a^{J_k}, \vec{a}^{I_m}) = Q_j^{\vec \pi,\, J_k\sqcup I_m}(s, \, \vec{a}^{J_k\sqcup I_m}) - Q_j^{\vec \pi,\, J_k}(s, \, \vec{a}^{J_k}).
    \end{equation}
    Finally, we define the multi-agent fair advantage function as:
    \begin{equation}
        A_F^{\vec\pi,\, I_m}(s,\vec a^{J_k}, \vec a^{I_m}) = \sumi[j]\frac{1}{\mathcal V_j^{\vec\pi}(s)}A_j^{\vec \pi,\, I_m}(s, \vec a^{J_k}, \vec{a}^{I_m})
    \end{equation}
\end{definition}

We now introduce a pivotal lemma which states that the fair advantage function can be split into a summation of each agent's local advantages.
\begin{restatable}{lemma}{advantagedecomp} \label{lem:advantage_decomp}
    For any agent j and subset $I_m\subseteq I_n$, we have 
    \begin{equation}
        A_F^{\vec \pi,\, I_m}(s, \vec{a}^{I_m}) = A_F^{\vec \pi,\, I_m}(s, a^{\void}, \vec{a}^{I_m})=\sum_{l=1}^m A_F^{\vec \pi,\, i_l}(s, \vec a^{I_{l-1}}, a^{i_l})
    \end{equation}
\end{restatable}

See Appendix \ref{sec:proof_lemma_3.2} for the proof. This lemma is the foundation for sequential update. Let $I_n$ be an arbitrary order. Suppose agent $i_1$ takes the first action $a^{i_1} $such that it improves its local fair advantage function, the subsequent agent $i_m$ can then sample its own action $a^{i_m}$ by taking into account the actions drawn by the preceding agents such that its action also improves its local fair advantage function -- then resulting into monotonic improvement. To formalize this intuition, we need the following definitions.
\begin{definition}
    \label{def:local_surrogate}
    Let $\vec\pi$ be a joint policy, $\vec\pi'^{I_{m-1}}$ be some other joint policy of agents $I_{m-1}$, and $\pi'^{i_m}$ be some other policy of agent $i_m$. Then
    \begin{equation}
        \mathcal L_{\vec \pi}^{I_m}(\vec{\pi}'^{I_{m-1}}, \pi'^{i_m})=\E_{s_\sim d_{\vec\pi}, \;\vec{a}^{I_{m-1}}\sim\vec{\pi}'^{I_{m-1}}(.|s), \, a^{i_m}\sim\pi'^{i_m}(.|s)}\left[A_F^{\vec\pi,\,i_m}\left(s, \vec a ^{I_{m-1}}, a^{i_m}\right)\right]
    \end{equation}
    Note that, for any $\vec\pi'^{I_{m-1}}$ and $\pi'^{i_m}=\pi^{i_m}$
    \begin{equation}
        \mathcal L_{\vec \pi}^{I_m}(\vec{\pi}'^{I_{m-1}}, \pi^{i_m})=0
    \end{equation}
\end{definition}

Thanks to lemma and theorem, we can generalize HATRL to our global fair objective.
\begin{restatable}{lemma}{surrogatedecomp}
    \label{lem:surrogate_decomp}
    Let $\vec\pi$ and $\vec\pi'$ be two joint policies, and $I_n$ an order, the surrogate function can be decomposed into $n$ summands:
    \begin{equation}
        \diffj\geq \sumi[m]\mathcal L_{\vec \pi}^{I_m}(\vec{\pi}'^{I_{m-1}}, \pi'^{i_m}) - C\cdot \dkl{\pi^{i_m}}{\pi'^ {i_m}}.
    \end{equation} 
\end{restatable}

See Appendix \ref{sec:proof_lemma_3.3} for the proof. This remarkable statement gives an idea about how a joint policy can be improved to maximize our global fair objective. Specifically, we can sequentially update agents policy by maximizing its corresponding summand $\mathcal L_{\vec \pi}^{I_m}(\vec{\pi}'^{I_{m-1}}, \pi'^{i_m}) - C\dkl{\pi^{i_m}}{\pi'^ {i_m}}$ which can always be positive as we can nullify the latter by making no policy update. Therefore, any policy updates that lead to a positive summand improves the overall summation. Hence, we propose the following algorithm.  

\begin{algorithm}[H]
\caption{Multi-Agent Policy Iteration with Monotonic Improvement Guarantee}
\label{algo:fhatrl}
\begin{algorithmic}[1] 
    \State \textbf{Initialize} the joint policy $\vec{\pi}_0 = (\pi_0^1, \dots, \pi_0^n)$.
    \For{$k = 0, 1, \dots$}
        \State Compute the fair advantage function $A_F^{\vec{\pi}_k}(s, \vec{a})$ for all state-action pairs $(s, \mathbf{a})$ and all agent $j$.
        \State Compute $\omega = \max_{j,s,\vec{a}} |A_j^{\vec{\pi}_k}(s, \vec{a})|$ and $C=\frac{4n\alpha \nu^{-1-\alpha}}{(1-\gamma)^2}\omega^2 + \frac{4n\omega\gamma}{\epsilon^\alpha(1-\gamma)^2}$.
        \State Draw a permutation $I_n$ of agents at random.
        \For{$m = 1, \dots, n$}
            \State Make an update $\pi_{k+1}^{i_m} = \arg \max_{\pi^{i_m}} \left[\mathcal L_{\vec \pi_k}^{I_m}(\vec{\pi}_{k+1}^{I_{m-1}}, \pi^{i_m}) - C D_{\text{KL}}(\pi_k^{i_m} \parallel \pi^{i_m}) \right]$.
        \EndFor
    \EndFor
\end{algorithmic}
\end{algorithm}

\subsection{Theoretical analysis}
In this section, we highlight the link between the surrogate function and the global fair objective. We first introduce a directional derivative definition.

\begin{restatable}{definition}{gateaux}\textit{(Gâteaux Derivative)}
    Let $\vec\pi,\vec\pi'$ be two joint policies and $\eta>0$ a positive number. Let us define the perturbed policies $\pi_\eta^i=\pi^i+\eta(\pi'^i-\pi^i)$, all together define the joint perturbed policy $\vec\pi_\eta$. 
    If the limit exists, we define the Gateaux derivative of a function $f(\vec\pi)$ as $\mathfrak D_{\vec\pi'}f(\vec\pi)=\lim_{\eta\rightarrow 0}\frac{f(\vec\pi_\eta)-f(\vec\pi)}{\eta}$. 
\end{restatable}

The surrogate function is actually the Gateaux derivative of $J(\vec\pi)$ in the $\vec\pi'$ direction.
\begin{restatable}{lemma}{surrogategateaux}
    \label{lem:surrogate_gateaux}
   Let $\vec\pi$ be a policy and $\vec\pi'$ be another policy. The following statement holds: $\mathfrak D_{\vec\pi'}J(\vec\pi)=\mathcal L_{\vec\pi}(\vec\pi')$.
\end{restatable}

See Appendix \ref{sec:proof_lemma_3.4} for the proof. The lemma provides a deeper understanding of Fair HATRL. Indeed, at each iteration $k+1$, the current joint policy $\vec\pi_k$ is moving toward the joint policy $\vec\pi'$ with the steepest slope to climb that remains in the trust region. Notice that similar result has be been shown for TRL.

Now we state the main properties enjoyed by Algorithm \ref{algo:fhatrl}. We first need to introduce the following solution concept.

\begin{definition}[Nash equilibrium]
    In a fully-cooperative game, a joint policy $\vec\pi_*=(\pi_*^1,\dots,\pi_*^n)$ is a Nash equilibrium (NE) if for every $i\in N$, $\piî\in\Pi^i$ implies $J(\vec\pi_*)\geq J(\pi^i, \vec \pi_*^{-i})$.
\end{definition}
NE (\cite{nash_equilibrium}) is common game-theory solution concept introduced by Nash. It describes an equilibrium where no agent can improve the global objective by deviating alone from its NE policy.

\begin{restatable}{theorem}{convergence}
    \label{thm:convergence}
    Supposing in Algorithm \ref{algo:fhatrl} any permutation of agents has a fixed non-zero probability to begin the update, the joint policies induced by the learning algorithm enjoy the following properties:
    \begin{itemize}[nosep]
        \item Monotonic Improvement property: $J(\vec\pi_{k+1})\geq J(\vec\pi_k)$,
        \item The fair global objective converges: $\lim_{k\rightarrow\infty}J(\vec\pi_k)=J_\infty$,
        \item The set of joint policies' limit points is non empty, each of which is a Nash equilibrium.
    \end{itemize}
\end{restatable}

See Appendix \ref{sec:proof_theorem} for the proof. The monotonic improvement and the convergence of the global fair objective is pretty straightforward. The characterization of the limit points provides some guarantee about the convergence of Fair-HATRL.


\section{Practical algorithms}

In practical settings, we employ Deep Reinforcement Learning (DRL) by parameterizing each agent’s policy $\pi^i$ with $\theta^i$ to accommodate high-dimensional state and action spaces. For brevity, we denote the policy $\pi^i_{\theta^i}$ simply as $\theta^i$ where the context is clear.
Fair-HATRL and HATRL share a core architectural framework; specifically, Fair-HATRL can be obtained from HATRL by substituting the standard multi-agent advantage function with our proposed global fair advantage function. Since both functions share the same analytical properties, we provide a high-level overview of the practical implementation below and refer the reader to \cite{hatrl} for further details.

\subsection{$\alpha$-fair HATRPO}
Similarly to HATRPO (\cite{hatrl}) and TRPO (\cite{trl}), we replace the KL-penalty by a KL-constraint to avoid computing $C$ and $\dkl{\theta^{i_m}_k}{\theta^{i_m}}$ in Algorithm \ref{algo:fhatrl}. Formally, at every iteration $k+1$, given a permutation $I_n$, Fair-HATRPO sequentially optimises agent $i_m$ policy parameters $\theta^{i_m}_{k+1}$ by solving the following optimization problem for a given a KL-radius $\delta>0$:
\begin{equation}
\begin{split}
    \theta^{i_m}_{k+1}=\underset{\theta^{i_m}}{\arg\max}\;\;\E_{s_\sim d_{\vec\theta_k}, \;\vec{a}^{I_{m-1}}\sim\vec{\theta}_{k+1}^{I_{m-1}}(.|s), \, a^{i_m}\sim\theta^{i_m}(.|s)}\left[A_F^{\vec\theta_k,\,i_m}\left(s, \vec a ^{I_{m-1}}, a^{i_m}\right)\right] \\
    \text{subject to } \underset{s\sim d_{\vec\theta_k}}{\E}\left[D_{\text{KL}}\left(\theta^{i_m}_k(\cdot|s),\theta^{i_m}(\cdot|s)\right)\right] \leq \delta. \qquad\qquad\quad
\end{split}
\end{equation}

One can solve this constrained optimization problem by using a linear approximation to the objective function and a quadratic approximation of the KL-constraint leading to the following update rule:
\begin{equation}
    \theta_{k+1}^{i_m} = \theta_k^{i_m} + \beta^j \sqrt{\frac{2\delta}{g_k^{i_m} (H_k^{i_m})^{-1} g_k^{i_m}}} (H_k^{i_m})^{-1} g_k^{i_m} ,
\end{equation}
where $H_k^{i_m} = \nabla_{\theta^{i_m}}^2 \mathbb{E}_{s \sim d_{\vec{\theta}_k}} \left[ \text{D}_{\text{KL}} \left( \theta_k^{i_m} (\cdot | s), \theta^{i_m} (\cdot | s) \right) \right] \Big|_{\theta^{i_m} = \theta_k^{i_m}}$ is the Hessian of the expected KL-divergence, $g_k^{i_m}$ is the gradient of the objective function, $\beta^j$ is computed via backtracking line search, and the product $(H_k^{i_m})^{-1}g_k^{i_m}$ can be efficiently computed using the conjugate gradient algorithm.

For a given batch $\mathcal B$ of trajectories with length $T$, let $\hat A^{\vec\theta_k}_F(s,\vec a)$ be the estimator of the fair advantage function $A^{\vec\theta_k}_F(s,\vec a)$. Then, the gradient of the objective function can be estimated using the following formula:
\begin{equation}
    \hat{g}_{k}^{i_{m}}=\frac{1}{|\mathcal{B}|} \sum_{\tau \in \mathcal{B}} \sum_{t=0}^{T} M_F^{I_m}\left(s_{t}, \vec{a}_{t}\right) \nabla_{\theta^{i_{m}}} \log \pi_{\theta^{i_{m}}}^{i_{m}}\left(a_{t}^{i_m} | s_{t}\right)|_{\theta^{i_{m}}=\theta_{k}^{i_{m}}} .
\end{equation}
where $M^{I_m}_F = \frac{\vec{\pi}'^{I_{m-1}}(\vec{a}^{I_{m-1}} | s)}{\vec{\pi}^{I_{m-1}}(\mathbf{a}^{I_{m-1}} | s)} \hat{A}_F^{\vec\theta_k}(s, \vec{a})$. 
The estimator of the fair advantage function is estimated from individual state-value functions' estimates (by the individual critic) and individual advantage functions' estimates (like GAE--\cite{gae}).  

\subsection{$\alpha$-fair HAPPO}
Similarly to HAPPO and PPO, we propose a first order approximation and avoid the computation of the hessian $H_k^{i_m}$. Thus, agent $i_m$ updates its parameter $\theta^{i_m}$ by maximizing the following clipping objective 
\begin{equation}
    \mathbb{E}_{s \sim d_{\vec\theta_k}, \vec a \sim \vec{\theta}_{k}}\left[\min \left(\frac{\pi_{\theta^{i m}}^{i_{m}}\left(a^{i_m} | s\right)}{\pi_{\theta_{k}^{i m}}^{i_{m}}\left(a^{i_m} | \mathrm{s}\right)} M^{I_m}(s, \vec a), \operatorname{clip}\left(\frac{\pi_{\theta^{i m}}^{i_{m}}\left(a^{i_m} | s\right)}{\pi_{\theta_{k}^{i_m}}^{i_{m}}\left(a^{i_m} | s\right)}, 1 \pm \epsilon\right) M^{I_m}(s, \vec a)\right)\right] .
\end{equation}


\section{Experiments and Results}
Finally, we test both our algorithms on two SSD -- CleanUp and Harvest (\cite{inequity_aversion}) -- against state-of-the-art algorithms. We built our experiments upon the environments provided by SocialJax (\cite{socialjax}). All hyperparameter settings and implementations details can be found in Appendix \ref{app:add_metrics}. Simulations were done on two RTX 4090 (24GB VRAM).
\subsection{Environments and Baselines}
\textbf{Harvest} presents a "tragedy of the commons" scenario where agents must balance their individual desire to collect apples with the need for collective restraint, as the resource only regenerates if a sufficient population of apples is maintained. In contrast, \textbf{CleanUp} models a public goods dilemma where agents face a conflict between harvesting apples and performing the labor-intensive task of cleaning a river; because apples only grow when the river is under a certain level of pollutants, agents must navigate the tension between contributing to the environment and free-riding on the efforts of others. 

We compare both FHATRPO and FHAPPO for different values of $\alpha\in\{0.5, 1,1.5\}$ and study the influence of that parameter on the efficiency and the fairness. We also test our algorithms against SOTA algorithms : HATRPO/HAPPO with global reward and FMAPPO with altruism level of 1 \cite{moi}. FMAPPO is based on the Proportional fairness objective and uses an PPO objective extrapolated from the gradient derivation with no theoretical guarantees. 

\subsection{Results}
We decide to do the comparison on two metrics : the total apples consumed (TAC) which measure the overall efficiency of the agents and the Gini index which measure how equal the utility distribution is. Other metrics can be found in Appendix \ref{app:add_metrics}.

\textbf{Common Harvest.} 
The Figure \ref{fig:result_common_harvest} presents the results obtained on Common Harvest. All the algorithms successfully learn the complexity of the environment. Interestingly, our FHATRPO and FHAPPO seem to perform slightly better than their baselines while reaching lower level of Gini index. Moreover, it seems that the higher $\alpha$ is, the fairer the outcome is except for $1.5$-FHAPPO which is probably due to the hard hyperparameter tuning. On the other hand, FMAPPO seems to perform pretty well. It levels FHATRPO on both TAC and Gini index and outperformed FHAPPO on pure efficiency. Nonetheless, FHAPPO reaches socially better result with a lower Gini index.

\def\makeginiharvest#1/#2/#3{
    \addplot[name path=#1-upper, draw=none, forget plot] 
        table [col sep=comma, x=TRAINING TIMESTEP, y=#1_MAX] {harvest_gini.csv};
        
    \addplot[name path=#1-lower, draw=none, forget plot] 
        table [col sep=comma, x=TRAINING TIMESTEP, y=#1_MIN] {harvest_gini.csv};
        
    \addplot[fill=#2!20, opacity=0.6, forget plot] 
        fill between[of=#1-upper and #1-lower];
        
    \addplot[thick, #2, mark=none] 
        table [col sep=comma, x=TRAINING TIMESTEP, y=#1_MEAN] {harvest_gini.csv};
        
    \addlegendentry{#3}
}

\def\makeapplesharvest#1/#2/#3{
    \addplot[name path=#1-upper, draw=none, forget plot] 
        table [col sep=comma, x=TRAINING TIMESTEP, y=#1_MAX] {harvest_apples.csv};
        
    \addplot[name path=#1-lower, draw=none, forget plot] 
        table [col sep=comma, x=TRAINING TIMESTEP, y=#1_MIN] {harvest_apples.csv};
        
    \addplot[fill=#2!20, opacity=0.6, forget plot] 
        fill between[of=#1-upper and #1-lower];
        
    \addplot[thick, #2, mark=none] 
        table [col sep=comma, x=TRAINING TIMESTEP, y=#1_MEAN] {harvest_apples.csv};
        
    \addlegendentry{#3}
}

\begin{figure}[htbp]
    \centering
    
    \resizebox{\textwidth}{!}{
        
        \begin{tikzpicture}
            \begin{groupplot}[
                group style={
                    group size=3 by 2,
                    horizontal sep=1cm, 
                    vertical sep=1.3cm      
                },
                try min ticks=5,
                legend pos=north west,
                legend cell align=left,
                legend image code/.code={
                    \draw[mark repeat=2,mark phase=2, #1] 
                        plot coordinates {(0cm,0cm) (0.15cm,0cm) (0.3cm,0cm)};
                },
                xlabel shift={-5pt},
                title style={font=\scriptsize, yshift=-8pt},         
                legend style={font=\fontsize{4}{4}\selectfont, row sep=-3pt, fill opacity=0.7,},   
                label style={font=\tiny}, 
                tick label style={font=\fontsize{4}{4}\selectfont},
                tick style={font=\fontsize{1}{4}\selectfont}, 
                width=6cm,  
                height=5cm,
                grid=major
            ]

                \nextgroupplot[title style={text width=3cm, align=center}, title={Total apples consumed HATRPO/FHATRPO}, xlabel={timestep}, ylabel={Apples consumed}]
                \pgfplotsinvokeforeach{HATRPO/red/{HATRPO}, FHATRPO_0.5/violet/{$\alpha$=0.5}, FHATRPO_1/cyan/{$\alpha$=1}, FHATRPO_1.5/magenta/{$\alpha$=1.5}}{
                    \makeapplesharvest#1
                }
                
                \nextgroupplot[title style={text width=3cm, align=center}, title={Total apples consumed HAPPO/FHAPPO}, xlabel={timestep}, ylabel={Apples consumed}]
                \pgfplotsinvokeforeach{HAPPO/blue/{HAPPO}, FHAPPO_0.5/orange/{$\alpha$=0.5}, FHAPPO_1/purple/{$\alpha$=1}, FHAPPO_1.5/teal/{$\alpha$=1.5}}{
                    \makeapplesharvest#1
                }
                
                \nextgroupplot[title style={text width=3cm, align=center}, title={Total apples consumed for $\alpha=1$}, xlabel={timestep}, ylabel={Apples consumed}]
                \pgfplotsinvokeforeach{FMAPPO/brown/{FMAPPO}, FHATRPO_1/cyan/{FHATRPO}, FHAPPO_1/purple/{FHAPPO}}{
                    \makeapplesharvest#1
                }
                
                \nextgroupplot[title style={text width=3cm, align=center}, title={Gini index HATRPO/FHATRPO}, xlabel={timestep}, ylabel={Gini Index}]
                \pgfplotsinvokeforeach{HATRPO/red/{HATRPO}, FHATRPO_0.5/violet/{$\alpha$=0.5}, FHATRPO_1/cyan/{$\alpha$=1}, FHATRPO_1.5/magenta/{$\alpha$=1.5}}{
                    \makeginiharvest#1
                }
                
                \nextgroupplot[title style={text width=3cm, align=center}, title={Gini Index HAPPO/FHAPPO}, xlabel={timestep}, ylabel={Gini Index}]
                \pgfplotsinvokeforeach{HAPPO/blue/{HAPPO}, FHAPPO_0.5/orange/{$\alpha$=0.5}, FHAPPO_1/purple/{$\alpha$=1}, FHAPPO_1.5/teal/{$\alpha$=1.5}}{
                    \makeginiharvest#1
                }
                
                \nextgroupplot[title style={text width=3cm, align=center}, title={Gini Index for $\alpha=1$}, xlabel={timestep}, ylabel={Gini Index}]
                \pgfplotsinvokeforeach{FMAPPO_1/brown/{FMAPPO}, FHATRPO_1/cyan/{FHATRPO}, FHAPPO_1/purple/{FHAPPO}}{
                    \makeginiharvest#1
                }
                
            \end{groupplot}
        \end{tikzpicture}
        
    } 
    
    \caption{Results on Common Harvest. Each line is obtained by averaging its actual value on a rolling window of size 100 and its shaded area corresponds to its minimum and maximum on the same rolling window.}
    \label{fig:result_common_harvest}
\end{figure}
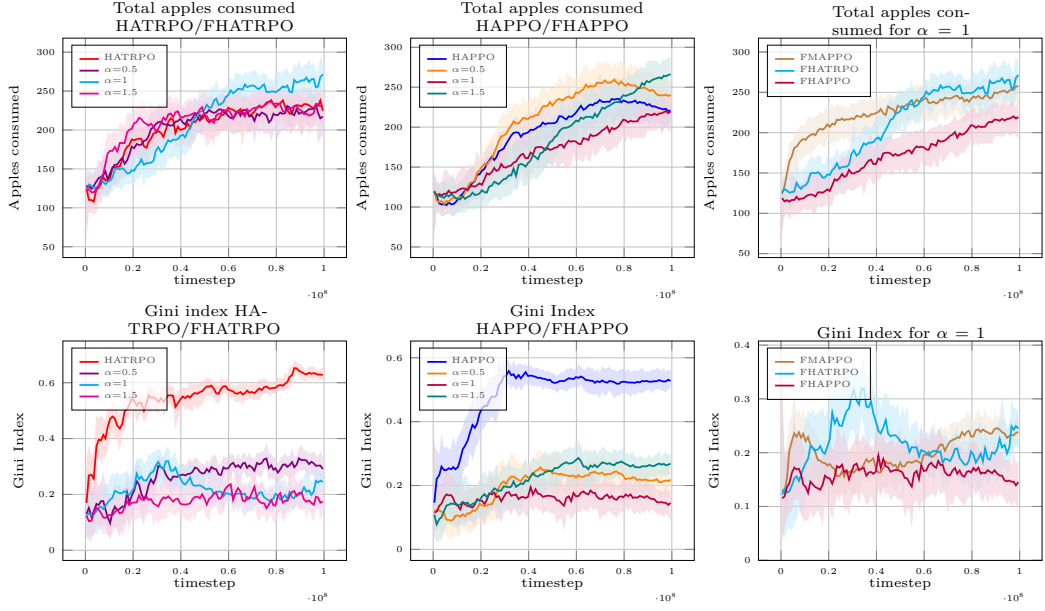

\textbf{CleanUp}
Results are presented in Figure \ref{fig:result_cleanup}. We can draw similar conclusion compared to Harvest ; especially that higher value of $\alpha$ leads to fairer outcomes. Nonetheless, the efficiency advantage of our fair algorithm is minimal and FMAPPO seems to be better than our fair algorithms. We also notice that FHAPPO struggles for $\alpha=1.5$ as the learning curve suffered from high drop of TAC.

\def\makeginicleanup#1/#2/#3{
    \addplot[name path=#1-upper, draw=none, forget plot] 
        table [col sep=comma, x=TRAINING TIMESTEP, y=#1_MAX] {cleanup_gini.csv};
        
    \addplot[name path=#1-lower, draw=none, forget plot] 
        table [col sep=comma, x=TRAINING TIMESTEP, y=#1_MIN] {cleanup_gini.csv};
        
    \addplot[fill=#2!20, opacity=0.6, forget plot] 
        fill between[of=#1-upper and #1-lower];
        
    \addplot[thick, #2, mark=none] 
        table [col sep=comma, x=TRAINING TIMESTEP, y=#1_MEAN] {cleanup_gini.csv};
        
    \addlegendentry{#3}
}

\def\makeapplescleanup#1/#2/#3{
    \addplot[name path=#1-upper, draw=none, forget plot] 
        table [col sep=comma, x=TRAINING TIMESTEP, y=#1_MAX] {cleanup_apples.csv};
        
    \addplot[name path=#1-lower, draw=none, forget plot] 
        table [col sep=comma, x=TRAINING TIMESTEP, y=#1_MIN] {cleanup_apples.csv};
        
    \addplot[fill=#2!20, opacity=0.6, forget plot] 
        fill between[of=#1-upper and #1-lower];
        
    \addplot[thick, #2, mark=none] 
        table [col sep=comma, x=TRAINING TIMESTEP, y=#1_MEAN] {cleanup_apples.csv};
        
    \addlegendentry{#3}
}

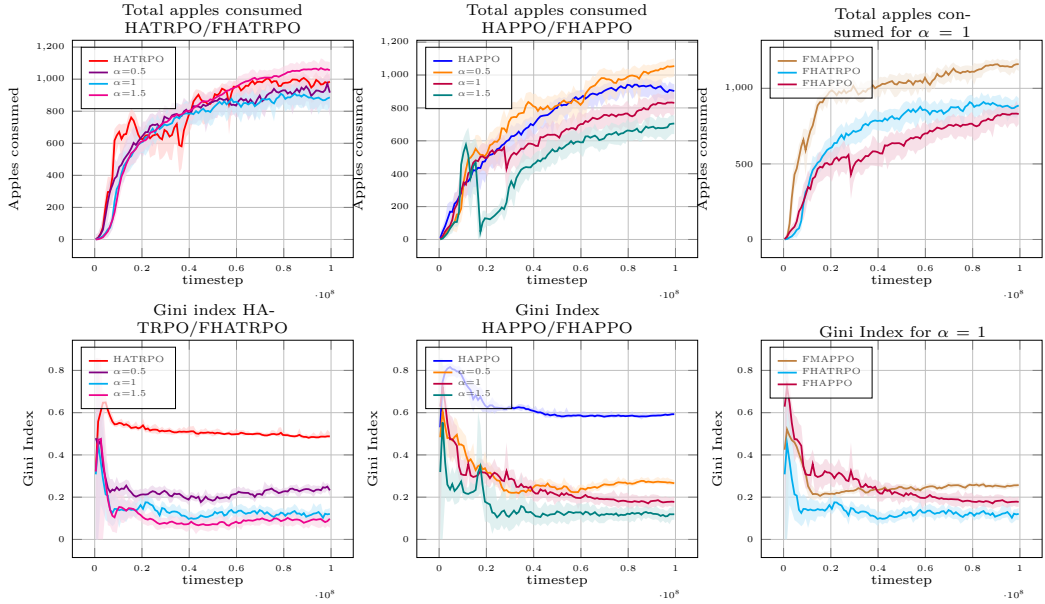
\begin{figure}[htbp]
    \centering
    
    \resizebox{\textwidth}{!}{
        
        \begin{tikzpicture}
            \begin{groupplot}[
                group style={
                    group size=3 by 2,
                    horizontal sep=1cm, 
                    vertical sep=1.3cm      
                },
                try min ticks=5,
                legend pos=north west,
                legend cell align=left,
                legend image code/.code={
                    \draw[mark repeat=2,mark phase=2, #1] 
                        plot coordinates {(0cm,0cm) (0.15cm,0cm) (0.3cm,0cm)};
                },
                title style={font=\scriptsize, yshift=-8pt},         
                xlabel shift={-5pt},
                legend style={font=\fontsize{4}{4}\selectfont, row sep=-3pt, fill opacity=0.7,},   
                label style={font=\tiny}, 
                tick label style={font=\fontsize{4}{4}\selectfont},
                tick style={font=\fontsize{1}{4}\selectfont}, 
                width=6cm,  
                height=5cm,
                grid=major
            ]

                \nextgroupplot[title style={text width=3cm, align=center}, title={Total apples consumed HATRPO/FHATRPO}, xlabel={timestep}, ylabel={Apples consumed}]
                \pgfplotsinvokeforeach{HATRPO/red/{HATRPO}, FHATRPO_0.5/violet/{$\alpha$=0.5}, FHATRPO_1/cyan/{$\alpha$=1}, FHATRPO_1.5/magenta/{$\alpha$=1.5}}{
                    \makeapplescleanup#1
                }
                
                \nextgroupplot[title style={text width=3cm, align=center}, title={Total apples consumed HAPPO/FHAPPO}, xlabel={timestep}, ylabel={Apples consumed}]
                \pgfplotsinvokeforeach{HAPPO/blue/{HAPPO}, FHAPPO_0.5/orange/{$\alpha$=0.5}, FHAPPO_1/purple/{$\alpha$=1}, FHAPPO_1.5/teal/{$\alpha$=1.5}}{
                    \makeapplescleanup#1
                }
                
                \nextgroupplot[title style={text width=3cm, align=center}, title={Total apples consumed for $\alpha=1$}, xlabel={timestep}, ylabel={Apples consumed}]
                \pgfplotsinvokeforeach{FMAPPO/brown/{FMAPPO}, FHATRPO_1/cyan/{FHATRPO}, FHAPPO_1/purple/{FHAPPO}}{
                    \makeapplescleanup#1
                }
                
                \nextgroupplot[title style={text width=3cm, align=center}, title={Gini index HATRPO/FHATRPO}, xlabel={timestep}, ylabel={Gini Index}]
                \pgfplotsinvokeforeach{HATRPO/red/{HATRPO}, FHATRPO_0.5/violet/{$\alpha$=0.5}, FHATRPO_1/cyan/{$\alpha$=1}, FHATRPO_1.5/magenta/{$\alpha$=1.5}}{
                    \makeginicleanup#1
                }
                
                \nextgroupplot[title style={text width=3cm, align=center}, title={Gini Index HAPPO/FHAPPO}, xlabel={timestep}, ylabel={Gini Index}]
                \pgfplotsinvokeforeach{HAPPO/blue/{HAPPO}, FHAPPO_0.5/orange/{$\alpha$=0.5}, FHAPPO_1/purple/{$\alpha$=1}, FHAPPO_1.5/teal/{$\alpha$=1.5}}{
                    \makeginicleanup#1
                }
                
                \nextgroupplot[title style={text width=3cm, align=center}, title={Gini Index for $\alpha=1$}, xlabel={timestep}, ylabel={Gini Index}]
                \pgfplotsinvokeforeach{FMAPPO/brown/{FMAPPO}, FHATRPO_1/cyan/{FHATRPO}, FHAPPO_1/purple/{FHAPPO}}{
                    \makeginicleanup#1
                }
                
            \end{groupplot}
        \end{tikzpicture}
        
    } 
    
    \caption{Results on CleanUp. Each line is obtained by averaging its actual value on a rolling window of size 100 and its shaded area corresponds to its minimum and maximum on the same rolling window.}
    \label{fig:result_cleanup}
\end{figure}

\section{Discussions and limitations}
During hyperparameter optimization, we observed significant performance degradation for high values of $\alpha$ (specifically, $\alpha > 1$). As $\alpha$ increases, the fair advantage function (Equation \ref{equ:fair_advantage_function}) effectively diminishes the weights of agents performing above the collective average, causing their corresponding terms in the summation to vanish. Consequently, the global objective becomes heavily biased toward the lowest-performing agents. This creates a high sensitivity to the "worst learners"; if these agents fail to explore effectively or hit a local optimum, the resulting gradients in the fair advantage function become volatile, leading to training instability. Using smaller KL-constraint/PPO-clipping can partially solve this issue.

On the other hand, the $\alpha$-fair HATRL framework suffers from several limitations. Namely, it relies on a strong theoretical constraint : positive bounded reward which narrows down the usable environments. Moreover, $\alpha$-fair HATRPO and $\alpha$-fair HAPPO require fully-observable state which is rarely the case in real world situation. Approximation using observation-based estimator with Recurent Neural Network can mitigate this but lack the theoretical properties.

\section{Conclusion} 
This work introduces a novel framework that bridges $\alpha$-fairness with Heterogeneous-Agent Trust Region Learning (HATRL), successfully filling the gap between equitable objective functions and theoretically grounded learning. We proved that our approach provides monotonic improvement and convergence toward Nash Equilibria in Positive Markov games. Practically, our developed algorithms—$\alpha$-fair HATRPO and $\alpha$-fair HAPPO—demonstrated in sequential social dilemmas like CleanUp and CommonHarvest that they can achieve significantly better fairness (lower Gini indices) while maintaining high utilitarian efficiency. This framework establishes a mathematically safe foundation for training agents to adopt sustainable, pro-social behaviors in complex multi-agent systems.
For future work, we suggest to investigate mechanism to enhance exploration like entropy normalization or adaptive KL-constraint/PPO-clipping that freezes the best learners and improves the worst.

\begin{ack}
Use unnumbered first level headings for the acknowledgments. All acknowledgments
go at the end of the paper before the list of references. Moreover, you are required to declare
funding (financial activities supporting the submitted work) and competing interests (related financial activities outside the submitted work).
More information about this disclosure can be found at: \url{https://neurips.cc/Conferences/2026/PaperInformation/FundingDisclosure}.

Do {\bf not} include this section in the anonymized submission, only in the final paper. You can use the \texttt{ack} environment provided in the style file to automatically hide this section in the anonymized submission.
\end{ack}

\vfill
\clearpage
\newpage
\bibliographystyle{plain} 
\bibliography{NeurIPS}

\newpage
\appendix
\section{Preliminaries}

\subsection{Assumptions and Definitions}
We use the same assumption as in HATRL \cite{hatrl}:
\begin{assumption}
    \label{assumption}
    There exists $\eta \in \mathbb{R}$, such that $0 < \eta \ll 1$, and for every agent $i \in \mathcal{N}$, the policy space $\Pi^i$ is $\eta$-soft; that means that for every $\pi^i \in \Pi^i$, $s \in \mathcal{S}$, and $a^i \in \mathcal{A}^i$, we have $\pi^i(a^i|s) \geq \eta$.
\end{assumption}

\begin{definition}
    Let $X$ be a finite set and $p: X \rightarrow \mathbb{R}$, $q: X \rightarrow \mathbb{R}$ be two maps. Then, the notion of \textbf{distance} between $p$ and $q$ that we adopt is given by $\|p - q\| \triangleq \max_{x \in X} |p(x) - q(x)|$.
\end{definition}

We also introduce the new probability distribution for the contextualized state:
\begin{definition}
    Let $\tilde{\mathcal M}=(N, \tilde{\mathcal{S}}, \allowbreak  \{\mathcal{A}_i\}_{i \in N},   \tilde P, \tilde \rho_0, \{\tilde r_i\}_{i \in N}, \gamma)$ be a Markov game. We define the Markov game $\mathcal M=(N, \mathcal{S}, \allowbreak  \{\mathcal{A}_i\}_{i \in N},   P, \rho_0, \{r_i\}_{i \in N}, \gamma)$ for the contextualized state $\mathcal S=\tilde{\mathcal S}\times\tilde{\mathcal S}$ where for any contextualized state $s=(\tilde s, \tilde c)$:
    \begin{itemize}
        \item $\rho_0(s)=\rho_0(\tilde s, \tilde c)=\tilde \rho_0(\tilde s)\mathds{1}_{\{\tilde s=\tilde c\}}$ is the initial state distribution,
        \item for any joint action $\vec a\in \mathcal A$ and other state $s'$, $P(s,\vec a, s')=\tilde P(\tilde s, \vec a, \tilde s')\mathds{1}_{\{\tilde c'=\tilde c\}}$ is the transition kernel,
        \item for any agent $i$ and joint action $\vec a$, $r_i(s, \vec a)=\tilde r_i(\tilde s,\vec a)$ is the reward function.
    \end{itemize}
\end{definition}

\subsection{Preliminary results}
\begin{lemma} \label{lem:policy_space_topology}
    Every agent $i$’s policy space $\Pi_i$ is convex and compact under the maximum norm.
\end{lemma}

\begin{lemma} \label{lem:state_dist_continuity}
    The improper state distribution $d_{\vec\pi}$ is continuous in $\vec\pi$. 
\end{lemma}

\begin{lemma} \label{lem:q_continuity}
    Let $\vec \pi$ be a policy. Then, for any agent $i$, $Q_i^{\vec\pi}(s,\vec a)$ is  Lipsichtz-continuous in $\vec\pi$.
\end{lemma}
\begin{proof}
    For proof of the lemmas \ref{lem:policy_space_topology}, \ref{lem:state_dist_continuity}, \ref{lem:q_continuity}, see \cite{hatrl} (Appendix A.2 Lemma 3, 4, 5). Notice that in \cite{hatrl}, the authors prove the lemmas for a global reward function but the proof holds for any reward functions. The proof relies on the Bellman Equations for $Q_\pi$ (see \cite{sutton}, Chapter 3.5) which also hold in our Multi-agent settings.  
\end{proof}

\begin{lemma} \label{lem:alpha_fairness_lip}
    The function $U_\alpha$ is Lipschitz-continuous on $\mathbb R^+$.
\end{lemma}
\begin{proof}
    Let us recall $U_\alpha(x)=\frac{x^{1-\alpha}}{1-\alpha}$ for $\alpha\neq 1$ and $U_\alpha(x)=\ln(x)$ for $\alpha=1$.
    
    The $\alpha$-fairness function is differentiable on $[\nu;+\infty[$ and for any $\alpha\geq 0$,
    \begin{equation*}
        U_\alpha'(x)=x^{-\alpha}
    \end{equation*}

    On $[\nu;+\infty[$, $\left|U_\alpha'(x)\right|$ is bounded by $\epsilon^{-\alpha}$. Thus, $U_\alpha$ is Lipschitz-continuous with constant $\nu^{-\alpha}$. 

\end{proof}

\begin{corollary} \label{cor:value_function_continuity}
    From Lemma \ref{lem:q_continuity} and Lemma \ref{lem:alpha_fairness_lip}, for any agent $i$, we obtain that the following functions are Lipschitz-continuous in $\pi$:
    \begin{itemize}
    \item the state value function $V_i^{\vec\pi}(s) = \sum_{\vec a} \vec\pi(\vec a|s)Q_i^{\vec\pi}(s, \vec a)$,
    \item the advantage function $A_i^{\vec\pi}(s, \vec a) = Q_i^{\vec\pi}(s, \vec a) - V_i ^{\vec\pi}(s)$,
    \item the fair global objective $J(\vec\pi) = \mathbb{E}_{s \sim \rho_0} [\sumi[i] U_\alpha\left(\nu + V_i^{\vec\pi}(s)\right)]$,
    \item and the fair advantage function $A_F^{\vec{\pi}}(s,\vec a)=\sumi[j]\frac{A_j^{\vec\pi}(s,\vec a)}{\mathcal V_j^{\vec\pi}(s)}$
    \end{itemize}
\end{corollary}

\begin{lemma}
\label{lem:continuity_expected_fair_adv}
    Let $\vec\pi$ and $\vec\pi'$ be some policies. The quantity $\E_{s\sim d_{\vec\pi}, \vec a\sim \vec \pi'(\cdot|s)}\left[A_F^{\vec{\pi}}(s,\vec a)\right]$ is continuous with $\vec\pi$.
\end{lemma}

\begin{proof}
    We have 
    \begin{equation}
        \E_{s\sim d_{\vec\pi}, \vec a\sim \vec \pi'(\cdot|s)}\left[A_F^{\vec{\pi}}(s,\vec a)\right]=\sum_s\sum_{\vec a}d_{\vec\pi}(s)\vec\pi(\vec a|s) A_F^{\vec{\pi}}(s,\vec a),
    \end{equation}
    the continuity is induced by continuity of $d_{\vec{\pi}}$ (Lemma \ref{lem:state_dist_continuity}) and $A_F^{\vec{\pi}}$ (Corollary \ref{cor:value_function_continuity}).
\end{proof}

\section{Heterogeneous Agent Trust Region Learning with fair objective}
In this section, the goal is to build a Trust Region Learning framework for the global objective we defined previously. 
\subsection{Previous results}
In this section, we list the useful lemma proved in TRL \cite{trl} that holds in our multi-agent settings. 
\begin{lemma}\label{lem:diff_v0}
    Let $\vec\pi$ and $\vec\pi'$ be two stochastic policies' profiles. For any agent $i$, the difference in their expected returns for a given initial state $s_0\in \mathcal{S}$ is given by:
    
    \begin{equation*}
        V_i^{\vec{\pi}'}(s_0) - V_i^{\vec{\pi}}(s_0) = \mathbb{E}_{\tau \sim \vec{\pi}'} \left[ \sum_{t=0}^{\infty} \gamma^t A_i^{\vec{\pi}}(s_t, \vec{a}_t) \mid s_0 \right]
    \end{equation*}
\end{lemma}
\begin{proof}
    The proof is similar to the performance difference lemma demonstrated in \cite{trl} (Lemma 1) or in \cite{kakade} (Lemma 6.1):
    \begin{align*}
        \mathbb{E}_{\tau \sim \vec{\pi}'} \left[ \sum_{t=0}^{\infty} \gamma^t A_i^{\vec{\pi}}(s_t, \vec{a}_t) \mid s_0 \right]
        &= \mathbb{E}_{\tau \sim \vec{\pi}'} \left[ \sum_{t=0}^{\infty} \gamma^t \left( r_i(s_t, \vec{a}_t) + \gamma V_i^{\vec{\pi}}(s_{t+1}) - V_i^{\vec{\pi}}(s_t)  \right) \mid s_0 \right] \\ 
        &= \mathbb{E}_{\tau \sim \vec{\pi}'} \left[ \sum_{t=0}^{\infty} \gamma^t r_i(s_t, \vec{a}_t) \mid s_0 \right] \\
        &\qquad\qquad\qquad+ \mathbb{E}_{\tau \sim \vec{\pi}'} \left[ \sum_{t=0}^{\infty} \gamma^t \left( \gamma V_i^{\vec{\pi}}(s_{t+1}) - V^{\vec{\pi}}(s_t) \right) \mid s_0 \right] \\
        &= V_i^{\vec{\pi}'}(s_0) + \mathbb{E}_{\tau \sim \vec{\pi}'} \left[ \sum_{t=0}^{\infty} \left( \gamma^{t+1} V_i^{\vec{\pi}}(s_{t+1}) - \gamma^t V_i^{\vec{\pi}}(s_t) \right) \mid s_0 \right] \\
        &= V_i^{\vec{\pi}'}(s_0) + \mathbb{E}_{\tau \sim \vec{\pi}'} \left[ -V_i^{\vec{\pi}}(s_0) + \lim_{T \to \infty} \gamma^T V_i^{\vec{\pi}}(s_T) \mid s_0 \right]\\
        &= V_i^{\vec{\pi}'}(s_0) - V_i^{\vec{\pi}}(s_0)
    \end{align*}
\end{proof}

\begin{lemma} \label{lem:expected_advantage_null}
    For any agent $i$, any state $s \in \mathcal{S}$ and a joint strategy profile $\vec{\pi}$, the following result holds:
    \begin{equation*}
        \mathbb{E}_{\vec{a} \sim \vec{\pi}(\cdot|s)} [A_i^{\vec{\pi}}(s, \vec{a})] = 0
    \end{equation*}
\end{lemma}

\begin{proof} For any state $s$, we have
\begin{align*}
    \mathbb{E}_{\vec{a} \sim \vec{\pi}(\cdot|s)} [A_i^{\vec{\pi}}(s, \vec{a})]
    &=\mathbb{E}_{\vec{a} \sim \vec{\pi}(\cdot|s)} [Q_i^{\vec{\pi}}(s, \vec{a}) - V_i^{\vec{\pi}}(s)] \\
    &= \mathbb{E}_{\vec{a} \sim \vec{\pi}(\cdot|s)} [Q_i^{\vec{\pi}}(s, \vec{a})] - \mathbb{E}_{\vec{a} \sim \vec{\pi}(\cdot|s)} [V_i^{\vec{\pi}}(s)]\\
    &=V_i^{\vec{\pi}}(s) - V_i^{\vec{\pi}}(s) \\
    &= 0
\end{align*}
\end{proof}

\begin{lemma} \label{lem:borne_sup_advantage}
    Let us define $\bar{A}_i^{\vec{\pi}',\, \vec{\pi}}(s)=\mathbb{E}_{\vec{a} \sim \vec{\pi}'(\cdot|s)} \left[A_i^{\vec{\pi}}(s, \vec{a})\right]$.
    For any state $s \in \mathcal{S}$ and any $\kappa$-coupled policy pair $(\vec\pi,\vec\pi')$, the magnitude of the expected joint advantage function $A^{\vec{\pi}}(s, \vec{a})$ evaluated under the policy $\vec{\pi}'$ is upper-bounded by the Total Variation (TV) distance between the two profiles:
    \begin{equation*}
        \left| \mathbb{E}_{\vec{a} \sim \vec{\pi}'(\cdot|s)} [A_i^{\vec{\pi}}(s, \vec{a})] \right| \leq 2 D_{\text{TV}}\left(\vec{\pi}'(\cdot|s) \parallel \vec{\pi}(\cdot|s)\right) \max_{\vec{a}} |A_i^{\vec{\pi}}(s, \vec{a})|
    \end{equation*}
\end{lemma}

\begin{proof}
    We use the previous Lemma \ref{lem:expected_advantage_null} to write $0$ differently:
    \begin{align*}
        \left|\bar{A}_i^{\vec{\pi}',\, \vec{\pi}}(s)\right|
        &=\left| \mathbb{E}_{\vec{a} \sim \vec{\pi}'(\cdot|s)} \left[A_i^{\vec{\pi}}(s, \vec{a})\right] \right| \\
        &= \left| \mathbb{E}_{\vec{a} \sim \vec{\pi}'(\cdot|s)} [A_i^{\vec{\pi}}(s, \vec{a})] - \mathbb{E}_{\vec{a} \sim \vec{\pi}(\cdot|s)} [A_i^{\vec{\pi}}(s, \vec{a})] \right| \\
        &= \left| \sum_{\vec{a}} \left( \vec{\pi}'(\vec{a}|s) - \vec{\pi}(\vec{a}|s) \right) A_i^{\vec{\pi}}(s, \vec{a}) \right| \\
        &\leq \sum_{\vec{a}} \left| \vec{\pi}'(\vec{a}|s) - \vec{\pi}(\vec{a}|s) \right| \cdot \max_{\vec{a}} |A_i^{\vec{\pi}}(s, \vec{a})| \\
        &\leq 2 \epsilon_i(s) D_{\text{TV}}(\vec{\pi}'(\cdot|s) \parallel \vec{\pi}(\cdot|s)),
    \end{align*}
    where $\epsilon_i(s)=\max_{\vec a}|A_i^{\vec{\pi}}(s, \vec{a})|$.
\end{proof}

\begin{lemma} \label{lem:diff_advantage}
    Let $\kappa=\dkl{\vec\pi}{\vec\pi'}=\max_sD_{\text{KL}}\left(\vec\pi(\cdot|s), \vec\pi'(\cdot|s)\right)$
    Let $(\vec{\pi}, \vec{\pi}')$ be an $\kappa$-coupled policy pair. Then, the following statement holds:
    \begin{align*}
        \left| \mathbb{E}_{s_t \sim \vec{\pi}'} [A_i^{\vec{\pi}',\, \vec{\pi}}(s_t)] - \mathbb{E}_{s_t \sim\vec\pi} [A_i^{\vec{\pi}',\, \vec{\pi}}
        (s_t)] \right|
        &\leq 4\kappa \max_{s,a} |A_i^{\vec\pi}(s, a)| \\
        &\le 4\kappa(1 - (1 - \kappa)^t) \max_{s,a} |A_i^{\vec{\pi}}(s, a)|
    \end{align*}
\end{lemma}

\begin{proof}
    Let $\kappa=\dkl{\vec\pi}{\vec\pi'}$, then we can define an $\kappa$-coupled policy pair $(\vec\pi,\vec\pi')$ such that 
    \begin{equation*}
        \mathbb{P}(\vec{a}_k \neq \vec{a}'_k | s_k) \le \kappa.
    \end{equation*}
    See \cite{markovchain} (proposition 4.7).
    Hence, the probability that the two policies disagree on an action at any given time step is bounded by $\kappa$.

    Let $n_t$ be the number of times the policies diverge before time step $t$. The probability that they agree on \textit{all} steps up to $t-1$ is at least $(1 - \kappa)^t$. 
    \begin{equation*}
        \mathbb{E}_{s_t \sim \vec{\pi}'} [A_i^{\vec{\pi}',\, \vec{\pi}}(s_t)] = \mathbb P(n_t = 0) \mathbb{E}_{s_t \sim \vec{\pi}' | n_t = 0} [A_i^{\vec{\pi}',\, \vec{\pi}}(s_t)] + \mathbb P(n_t > 0) \mathbb{E}_{s_t \sim \vec{\pi}' | n_t > 0} [A_i^{\vec{\pi}',\, \vec{\pi}}(s_t)]
    \end{equation*}
    
    Similarly for the policy $\vec{\pi}$:

    \begin{equation*}
        \mathbb{E}_{s_t \sim \vec{\pi}} [A_i^{\vec{\pi}',\, \vec{\pi}}(s_t)] = \mathbb P(n_t = 0) \mathbb{E}_{s_t \sim \vec{\pi} | n_t = 0} [A_i^{\vec{\pi}',\, \vec{\pi}}(s_t)] + \mathbb P(n_t > 0) \mathbb{E}_{s_t \sim \vec{\pi} | n_t > 0} [A_i^{\vec{\pi}',\, \vec{\pi}}(s_t)]
    \end{equation*}

    Note that $n_t=0$ terms are equal as it means that $\vec{\pi}$ and $\vec{\pi}'$ agreed on all the timesteps less than $t$:
    \begin{equation*}
         \mathbb{E}_{s_t \sim \vec{\pi} | n_t = 0} [A_i^{\vec{\pi}',\, \vec{\pi}}(s_t)]= \mathbb{E}_{s_t \sim \vec{\pi}' | n_t = 0} [A_i^{\vec{\pi}',\, \vec{\pi}}(s_t)]
    \end{equation*}

    By definition, the probability that they disagree at least once before time $t$ is bounded by:
    \begin{equation*}
        \mathbb P(n_t > 0) \le 1 - (1 - \kappa)^t
    \end{equation*}

    Then, we have
    
    \begin{align}
        |\mathbb{E}_{s_t \sim \vec{\pi}' | n_t > 0} [A_i^{\vec{\pi}',\, \vec{\pi}}(s_t)] - \mathbb{E}_{s_t \sim \pi | n_t > 0}& [A_i^{\vec{\pi}',\, \vec{\pi}}(s_t)]| \\
        &\leq |\mathbb{E}_{s_t \sim \tilde{\pi} | n_t > 0} [A_i^{\vec{\pi}',\, \vec{\pi}}(s_t)]| + |\mathbb{E}_{s_t \sim \pi | n_t > 0} [A_i^{\vec{\pi}',\, \vec{\pi}}(s_t)]| \\
        &\leq 4\kappa \max_{s,a} |A_i^{\vec\pi}(s, a)|
    \end{align}

    Therefore, using Lemma \ref{lem:borne_sup_advantage}, we have
    \begin{align*}
        \left| \mathbb{E}_{s_t \sim \vec{\pi}'} [A_i^{\vec{\pi}',\, \vec{\pi}}(s_t)] - \mathbb{E}_{s_t \sim\vec\pi} [A_i^{\vec{\pi}',\, \vec{\pi}}(s_t)] \right|
        &\le 2(1 - (1 - \kappa)^t) \max_s |\bar{A}_i^{\vec{\pi}',\, \vec{\pi}}(s)| \\
        &\le 4\kappa(1 - (1 - \kappa)^t) \max_{s,a} |A_i^{\vec{\pi}}(s, a)|
    \end{align*}
\end{proof}
\subsection{Proof of Lemma \ref{lem:fair_surrogate}}
\label{sec:proof_lemma_3.1}
\fairsurrogate*

\begin{proof}
    We are looking for a lower bound of $\diffj$.
    \begin{equation}
        \diffj = \E_{s_0\sim\rho_0}\left[\sumi\ua\left(\epsilon+\vf{'}(s_0)\right) - \ua\left(\epsilon+\vf{}(s_0)\right)\right]
    \end{equation}

    As $\ua$ is concave, one interesting property is by
    \begin{equation}
        \frac{y-x}{x^\alpha}\geq \ua(y)-\ua(x)\geq \frac{y-x}{y^\alpha}.
    \end{equation}

    Applying this property to the previous difference and using Lemma \ref{lem:diff_v0} gives the lower bound
    \begin{align}
        \diffj 
        &\geq \E_{s_0\sim\rho_0}\left[\frac{\vf{'}(s_0)-\vf{}(s_0)}{\calvf{'}(s_0)}\right] \\
        &= \E_{s_0\sim\rho_0}\left[\sumi \frac{1}{\calvf{'}(s_0)}\mathbb{E}_{\tau \sim \vec{\pi}'|s_0} \left[ \sum_{t=0}^{\infty} \gamma^t A_j^{\vec{\pi}}(s_t, \vec{a}_t) \right]\right]\\
        &= \mathbb{E}_{\tau \sim \vec{\pi}'} \left[ \sum_{t=0}^{\infty} \gamma^t \left(\sumi \frac{A_j^{\vec{\pi}}(s_t, \vec{a}_t)}{\calvf['](s_0)}\right) \right] \\
        &= \mathcal{F}_{\vec{\pi}}(\vec{\pi}'),
    \end{align}
    and the upper bound
    \begin{align*}
        \diffj
        &\leq \E_{s_0\sim\rho_0}\left[\frac{\vf{'}(s_0)-\vf{}(s_0)}{\calvf(s_0)}\right] \\
        &= \mathbb{E}_{\tau \sim \vec{\pi}'} \left[ \sum_{t=0}^{\infty} \gamma^t \left(\sumi \frac{A_j^{\vec{\pi}}(s_t, \vec{a}_t)}{\calvf['](s_0)}\right) \right] \\
        &= \mathcal L_{\vec\pi}(\vec\pi').
    \end{align*}

    Let us define $A_j^{\vec{\pi}', \vec \pi}(s)=\E_{\vec{a}\sim \vec{\pi}'(\cdot|s)}\left[A_j^{\vec\pi}(s, \vec a)\right]$. Using that notation and conditional expectation, we have 
    \begin{align}
        \mathcal{F}_{\vec{\pi}}(\vec{\pi}')= \mathbb{E}_{\tau \sim \vec{\pi}'} \left[ \sum_{t=0}^{\infty} \gamma^t \left(\sumi \frac{A_j^{\vec{\pi}', \vec \pi}(s_t)}{\calvf['](s_0)}\right) \right]
    \end{align}

    As $s_t=(\tilde s_t, \tilde s_0)$ carries the initial state, we can drop the initial state dependency in the denominator:

    \begin{align}
        \mathcal{F}_{\vec{\pi}}(\vec{\pi}')
        &=\mathbb{E}_{\tau \sim \vec{\pi}'} \left[ \sum_{t=0}^{\infty} \gamma^t \left(\sumi \frac{A_j^{\vec{\pi}', \vec \pi}(s_t)}{\calvf['](s_0)}\right) \right]\notag\\
        &= \sum_s\rho_0(s)\sumt\gamma^t \sum_{s'}\mathbb{P}_{\vec{\pi}'}(s_t=s'\,|\, s_0=s) \left(\sumi \frac{A_j^{\vec{\pi}', \vec \pi}(s')}{\calvf['](s_0)}\right) \notag\\
        &=\sum_{\tilde s,\,\tilde c}\rho_0(\tilde s)\mathds{1}_{\{\tilde s = \tilde c\}}\sumt\gamma^t \sum_{\tilde s', \tilde c'}\mathbb{P}_{\vec{\pi}'}(\tilde s_t=\tilde s'\,|\, \tilde s_0=\tilde s)\mathds{1}_{\{\tilde c' = \tilde c\}} \left(\sumi \frac{A_j^{\vec{\pi}', \vec \pi}(\tilde s', \tilde c')}{\calvf['](\tilde s, \tilde c)}\right) \notag\\
        &=\sum_{\tilde s,\,\tilde c}\rho_0(\tilde s)\mathds{1}_{\{\tilde s = \tilde c\}}\sumt\gamma^t \sum_{\tilde s', \tilde c'}\mathbb{P}_{\vec{\pi}'}(\tilde s_t=\tilde s'\,|\, \tilde s_0=\tilde s)\mathds{1}_{\{\tilde c' = \tilde c\}} \left(\sumi \frac{A_j^{\vec{\pi}', \vec \pi}(\tilde s', \tilde c')}{\calvf['](\tilde c', \tilde c')}\right) \notag\\
        &=\sum_{\tilde s,\,\tilde c}\rho_0(\tilde s)\mathds{1}_{\{\tilde s = \tilde c\}} \sum_{\tilde s', \tilde c'}\sumt\gamma^t\mathbb{P}_{\vec{\pi}'}(\tilde s_t=\tilde s'\,|\, \tilde s_0=\tilde s)\mathds{1}_{\{\tilde c' = \tilde c\}} \left(\sumi \frac{A_j^{\vec{\pi}', \vec \pi}(\tilde s', \tilde c')}{\calvf['](\tilde c', \tilde c')}\right) \notag\\
        &=\mathbb{E}_{s \sim d_{\vec{\pi}'}} \left[\sumi \frac{A_j^{\vec{\pi}', \vec \pi}(s)}{\calvf['](s)} \right]
    \end{align}

    Because the value function and the state distribution of the other policy $\vec \pi '$ are unknown, we want to bound $\mathcal{F}_{\vec{\pi}}(\vec{\pi}')$ by $\vec \pi$-depended functions. Therefore, let us consider the following function:
    \begin{equation}
        \mathcal{L}_{\vec{\pi}}(\vec{\pi}')=\mathbb{E}_{s \sim d_{\vec{\pi}}} \left[\sumi \frac{A_j^{\vec{\pi}', \vec \pi}(s)}{\calvf(s)} \right]
    \end{equation}

    We now want to bound the absolute difference of $|\mathcal{F}_{\vec{\pi}}(\vec{\pi}') - \mathcal{L}_{\vec{\pi}}(\vec{\pi}')|$:
    \begin{align}
        |\mathcal{F}_{\vec{\pi}}(\vec{\pi}') - \mathcal{L}_{\vec{\pi}}(\vec{\pi}')|
        &=\left|\mathbb{E}_{s \sim d_{\vec{\pi}'}} \left[\sumi \frac{A_j^{\vec{\pi}', \vec \pi}(s)}{\calvf['](s)} \right]-\mathbb{E}_{s \sim d_{\vec{\pi}}} \left[\sumi \frac{A_j^{\vec{\pi}', \vec \pi}(s)}{\calvf(s)} \right]\right| \\[2ex]
        &\leq \underbrace{\left|\mathbb{E}_{s \sim d_{\vec{\pi}'}} \left[\sumi \frac{A_j^{\vec{\pi}', \vec \pi}(s)}{\calvf['](s)} \right] -\mathbb{E}_{s \sim d_{\vec{\pi}'}} \left[\sumi \frac{A_j^{\vec{\pi}', \vec \pi}(s)}{\calvf(s)} \right]\right|}_{=\Delta_1}\nonumber\\[2ex]
        &\qquad + \underbrace{\left|\mathbb{E}_{s \sim d_{\vec{\pi}'}} \left[\sumi \frac{A_j^{\vec{\pi}', \vec \pi}(s)}{\calvf(s)} \right]-\mathbb{E}_{s \sim d_{\vec{\pi}}} \left[\sumi \frac{A_j^{\vec{\pi}', \vec \pi}(s)}{\calvf(s)} \right]\right|}_{=\Delta_2}
    \end{align}

    The first error $\Delta_1$ is due to a shift of state value functions between the two policies $\vec \pi$ and $\vec \pi '$:

    \begin{align*}
        \Delta_1
        &=\left|\mathbb{E}_{s \sim d_{\vec{\pi}'}} \left[\sumi \frac{A_j^{\vec{\pi}', \vec \pi}(s)}{\calvf['](s)} \right] -\mathbb{E}_{s \sim d_{\vec{\pi}'}} \left[\sumi \frac{A_j^{\vec{\pi}', \vec \pi}(s)}{\calvf(s)} \right]\right| \\
        &=\left|\mathbb{E}_{s \sim d_{\vec{\pi}'}} \left[\sumi \left(\frac{1}{\calvf['](s)}  -\frac{1}{\calvf(s)}\right) A_j^{\vec{\pi}', \vec \pi}(s) \right]\right|
    \end{align*}

    The function $f(v)=(\epsilon +v) ^{-\alpha}$ is Lipschitz continuous on $\mathbb R ^+$. Thus, the following statement holds for any $v,\, w\in \mathbb R^+$:
    \begin{equation}
        |f(v)-f(w)|\leq L_\alpha|v-w|,
    \end{equation}
    where $L_\alpha=\max_{v\in\mathbb R^+}\;|f'(v)| = \alpha \nu^{-1-\alpha}$ is a Lipschitz constant. 

    Therefore, we have 
    \begin{align*}
        \left|\frac{1}{\calvf['](s)}  -\frac{1}{\calvf(s)}\right|
        &= \left|\frac{1}{f(\vf['](s))}  -\frac{1}{f(\vf(s))}\right| \\
        &\leq L_\alpha\cdot |\vf['](s)-\vf(s)| \\
        &= L_\alpha\cdot\left|\mathbb{E}_{\tau \sim \vec{\pi}'} \left[ \sum_{t=0}^{\infty} \gamma^t A_j^{\vec{\pi}}(s_t, \vec{a}_t) | s_0=s \right]\right| \\
        &= L_\alpha\cdot\left|\mathbb{E}_{\tau \sim \vec{\pi}'} \left[ \sum_{t=0}^{\infty} \gamma^t A_j^{\vec{\pi}', \vec \pi}(s_t) | s_0=s \right]\right| \\
        &\leq L_\alpha\cdot\mathbb{E}_{\tau \sim \vec{\pi}'} \left[\sum_{t=0}^{\infty} \gamma^t \left|A_j^{\vec{\pi}', \vec \pi}(s_t)\right| \;|\; s_0=s \right] \\
        &\leq \frac{2L_\alpha}{1-\gamma}A_{\max} D_{\text{TV}}^{\max}(\vec{\pi},\,\vec \pi ')
    \end{align*}

    Plugging back the previous result into the first difference leads to:
    \begin{align}
        \Delta_1
        &=\left|\mathbb{E}_{s \sim d_{\vec{\pi}'}} \left[\sumi \left(\frac{1}{\calvf['](s)}  -\frac{1}{\calvf(s)}\right) A_j^{\vec{\pi}', \vec \pi}(s) \right]\right|\notag\\
        &\leq \mathbb{E}_{s \sim d_{\vec{\pi}'}} \left[\sumi \left|\frac{1}{\calvf['](s)}  -\frac{1}{\calvf(s)}\right| \left|A_j^{\vec{\pi}', \vec \pi}(s)\right| \right]\notag\\
        \text{using Lemma \ref{lem:borne_sup_advantage}, }&\\
        &\leq \mathbb{E}_{s \sim d_{\vec{\pi}'}} \left[\sumi \frac{4L_\alpha}{1-\gamma}A_{\max}^2 \kappa^2 \right]\notag\\
        &= \frac{4nL_\alpha}{(1-\gamma)^2}A_{\max}^2 \kappa^2,
    \end{align}

    where $\kappa=D_{\text{TV}}^{\max}(\vec{\pi},\,\vec \pi ')$.

    Next, we have to find an upper bound of the second difference $\Delta_2$ which is due to the shift of state visitation distribution:
    \begin{align*}
        \Delta_2
        &=\left|\mathbb{E}_{s \sim d_{\vec{\pi}'}} \left[\sumi \frac{A_j^{\vec{\pi}', \vec \pi}(s)}{\calvf(s)} \right]-\mathbb{E}_{s \sim d_{\vec{\pi}}} \left[\sumi \frac{A_j^{\vec{\pi}', \vec \pi}(s)}{\calvf(s)} \right]\right| \notag\\
        &=\left|\mathbb{E}_{\tau \sim \vec{\pi}'} \left[\sumt\gamma^t \left(\sumi \frac{A_j^{\vec{\pi}', \vec \pi}(s_t)}{\calvf(s_0)} \right)\right]-\mathbb{E}_{\tau \sim \vec{\pi}} \left[\sumt\left(\sumi \frac{A_j^{\vec{\pi}', \vec \pi}(s)}{\calvf(s)} \right)\right]\right|\notag\\
        &=\sumt\gamma^t\left|\mathbb{E}_{\tau \sim \vec{\pi}'} \left[\sumi \frac{A_j^{\vec{\pi}', \vec \pi}(s_t)}{\calvf(s_0)}\right]- \mathbb{E}_{\tau \sim \vec{\pi}} \left[\sumi \frac{A_j^{\vec{\pi}', \vec \pi}(s_t)}{\calvf(s_0)}\right]\right|\notag\\
        &=\sumt\gamma^t\left|\E_{s_0\sim\rho_0}\left[\mathbb{E}_{\tau \sim \vec{\pi}'|s_0} \left[\sumi \frac{A_j^{\vec{\pi}', \vec \pi}(s_t)}{\calvf(s_0)}\right]- \mathbb{E}_{\tau \sim \vec{\pi}|s_0} \left[\sumi \frac{A_j^{\vec{\pi}', \vec \pi}(s_t)}{\calvf(s_0)}\right]\right]\right|\notag\\
        &\leq\sumt\gamma^t\E_{s_0\sim\rho_0}\left[\sumi\frac{1}{\calvf(s_0)}\left|\mathbb{E}_{\tau \sim \vec{\pi}'|s_0} \left[ A_j^{\vec{\pi}', \vec \pi}(s_t)\right]- \mathbb{E}_{\tau \sim \vec{\pi}|s_0} \left[ A_j^{\vec{\pi}', \vec \pi}(s_t)\right]\right|\right]\\
    \end{align*}

    Let us recall Lemma \ref{lem:diff_advantage}:
    \begin{equation*}
        \left|\mathbb{E}_{\tau \sim \vec{\pi}'|s_0} \left[ A_j^{\vec{\pi}', \vec \pi}(s_t)\right]- \mathbb{E}_{\tau \sim \vec{\pi}|s_0} \left[ A_j^{\vec{\pi}', \vec \pi}(s_t)\right]\right| \leq 4\kappa(1-(1-\kappa)^t)A_{\max}.
    \end{equation*}

    Moreover, $\frac{1}{\calvf(s_0)}$ is upper bounded by $\nu^{-\alpha}$. Thus, we have 
    \begin{align}
        \Delta_2
        &\leq\sumt\gamma^t\E_{s_0\sim\rho_0}\left[\sumi\frac{1}{\calvf(s_0)}\left|\mathbb{E}_{\tau \sim \vec{\pi}'|s_0} \left[ A_j^{\vec{\pi}', \vec \pi}(s_t)\right]- \mathbb{E}_{\tau \sim \vec{\pi}|s_0} \left[ A_j^{\vec{\pi}', \vec \pi}(s_t)\right]\right|\right]\\
        &\leq \sumt \gamma^t \E_{s_0\sim\rho_0}\left[4nA_{\max}\nu^{-\alpha}(1-(1-\kappa)^t)\kappa\right]\\
        &\leq \sumt \gamma^t 4nA_{\max}\nu^{-\alpha}(1-(1-\kappa)^t)\kappa\\
        &\leq \frac{4nA_{\max}\kappa}{\nu^\alpha}\left(\frac{1}{1-\gamma}-\frac{1}{1-\gamma(1-\kappa)}\right)\\
        &\leq \frac{4nA_{\max}\gamma\kappa^2}{\nu^\alpha(1-\gamma)^2 }
    \end{align}

    Therefore, by combining both inequalities, we have :
    \begin{align}
        \left|\mathcal F_{\vec \pi}(\vec\pi')-\mathcal L_{\vec\pi}(\vec\pi')\right|
        &\leq \Delta_1+\Delta_2\notag\\
        &\leq \frac{4nL_\alpha}{(1-\gamma)^2}A_{\max}^2 \kappa^2 + \frac{4nA_{\max}\gamma}{\nu^\alpha(1-\gamma)^2 }\kappa^2
    \end{align}
    
    This concludes the proof.
\end{proof}
\label{sec:proof_lemma_3.2}
\subsection{Proof of Lemma \ref{lem:advantage_decomp}}
The Lemma \ref{lem:advantage_decomp} is induced by individual advantage functions decomposition:
\begin{lemma}
    \label{lem:ind_adv_decomp}
    For any agent j and subset $I_m\subseteq I_n$, we have 
    \begin{equation}
        A_j^{\vec \pi,\, I_m}(s, \vec{a}^{I_m}) = A^{\vec \pi,\, I_m}(s, a^{\void}, \vec{a}^{I_m})=\sum_{l=1}^m A_j^{\vec \pi,\, i_m}(s, \vec a^{I_{m-1}}, \vec{a}^{i_m})
    \end{equation}
\end{lemma}

\begin{proof}
    For any state $s$ and joint action $\vec a^{I_m}$, we have
    \begin{align}
        A_j^{\vec \pi,\, I_m}(s, \vec{a}^{I_m}) 
        &= A^{\vec \pi,\, I_m}(s, a^{\void}, \vec{a}^{I_m})\\
        &=Q_j^{\vec \pi,\, \void\sqcup I_m}(s, \, \vec{a}^{\void\sqcup I_m}) - Q_j^{\vec \pi,\, \void}(s, \, a^{\void})\\
        &=Q_j^{\vec \pi,\, I_l}(s, \, \vec{a}^{I_l}) - V_j^{\vec \pi}(s)\\
        &=\sum_{l=1}^m Q_j^{\vec \pi,\, I_l}(s, \, \vec{a}^{I_l})-Q_j^{\vec \pi,\, I_{l-1}}(s, \, \vec{a}^{I_{l-1}})\\
        &=\sum_{l=1}^m A_j^{\vec \pi,\, i_l}(s, \vec a^{I_{l-1}}, a^{i_l})
    \end{align}
\end{proof}

\advantagedecomp*
\begin{proof}
    The proof is pretty straightforward by considering the Lemma \ref{lem:ind_adv_decomp}. For any set $I_m$, state $s$ and joint action $\vec a^{I_m}$, we have by definition
    \begin{align*}
        A_F^{\vec\pi,\, I_m}(s, \vec a^{I_m}) 
        &= \sumi[j]\frac{1}{\mathcal V_j^{\vec\pi}(s)}A_j^{\vec \pi,\, I_m}(s, \vec{a}^{I_m}) \\
        &= \sumi[j]\frac{1}{\mathcal V_j^{\vec\pi}(s)}\sum_{l=1}^m A_j^{\vec \pi,\, i_l}(s, \vec a^{I_{l-1}}, \vec{a}^{i_l})\\
        &= \sum_{l=1}^m \left(\sumi[j]\frac{1}{\mathcal V_j^{\vec\pi}(s)} A_j^{\vec \pi,\, i_l}(s, \vec a^{I_{l-1}}, a^{i_l})\right)\\
        &=\sum_{l=1}^m A_F^{\vec \pi,\, i_l}(s, \vec a^{I_{l-1}}, a^{i_l})
    \end{align*}
\end{proof}

\subsection{Proof of Lemma \ref{lem:surrogate_decomp}}
\label{sec:proof_lemma_3.3}
To prove Lemma \ref{lem:surrogate_decomp}, we first need to decompose the KL-penalty term:
\begin{lemma}
    \label{lem:dkl_decomp}
    For any joint policies $\vec \pi$ and $\vec \pi'$ and state $s$, the following statement holds 
    \begin{equation*}
        D_{\text{KL}}\left(\vec \pi(.|s) \; ||\; \vec \pi ' (.|s)\right)= \sumi D_{\text{KL}}\left(\pi_j(.|s) \; ||\; \pi_j (.|s)\right).
    \end{equation*}
\end{lemma}

\begin{proof}
    For any state $s$, we have 
    \begin{align*}
        D_{\text{KL}}\left(\vec{\pi}(\cdot|s) \parallel \vec{\pi}'(\cdot|s)\right) 
        &= \E_{\vec{a}\sim\vec{\pi}(.|s)}  \left[\log \vec{\pi}(\vec{a}|s) - \log\vec{\pi}'(\vec{a}|s)\right] \\
        &= \E_{\vec{a}\sim\vec{\pi}(.|s)}  \left[\log \prod_{j=1}^n \pi_j(a_j|s) - \log\prod_{j=1}^N \pi_j'(a_j|s)\right] \\
        &= \E_{\vec{a}\sim\vec{\pi}(.|s)}  \left[\sumi\log \pi_j(a_j|s) - \log \pi_j'(a_j|s)\right] \\
        &= \sumi\E_{a_j\sim\pi_j(.|s)}  \left[\log \pi_j(a_j|s) - \log \pi_j'(a_j|s)\right] \\
        &= \sum_{j=1}^n D_{\text{KL}}\left(\pi_j(\cdot|s) \parallel \pi'_j(\cdot|s)\right)
\end{align*}
\end{proof}

\surrogatedecomp*
\begin{proof}
    The proof relies on Lemma \ref{lem:advantage_decomp} and Lemma \ref{lem:dkl_decomp}:
    \begin{align}
        \diffj
        &\geq \mathcal L_{\vec\pi}(\vec\pi')-C\cdot D_{\text{KL}}^{\max}(\vec\pi\;||\; \vec\pi') \\
        &\geq \mathbb{E}_{s \sim d_{\vec{\pi}}, \,\vec a \sim \vec\pi'(.|s)} \left[\sumi \frac{A_j^{\vec \pi}(s, \vec a)}{\calvf(s)} \right] -C\cdot D_{\text{KL}}^{\max}(\vec\pi\;||\; \vec\pi') \\
        &\geq \mathbb{E}_{s \sim d_{\vec{\pi}}, \,\vec a \sim \vec\pi'(.|s)} \left[\sumi \frac{1}{\calvf(s)}\sumi[m] A_j^{\vec \pi,\, i_m}(s, \vec a^{I_{m-1}}, a^{i_m}) \right] \\
        &\qquad\qquad\qquad\qquad\qquad-\sumi[m]C\cdot D_{\text{KL}}^{\max}(\pi_{i_m}\;||\; \pi'_{i_m}) \\
        &\geq \sumi[m] \mathbb{E}_{s_\sim d_{\vec\pi}, \;\vec{a}^{I_{m-1}}\sim\vec{\pi}^{I_{m-1}}(.|s), \, a^{i_m}\sim\pi'^{i_m}(.|s)} \left[\sumi \frac{1}{\calvf(s)} A_j^{\vec \pi,\, i_m}(s, \vec a^{I_{m-1}}, a^{i_m}) \right] \nonumber\\
        &\qquad\qquad\qquad\qquad\qquad -C\cdot D_{\text{KL}}^{\max}(\pi_{i_m}\;||\; \pi'_{i_m}) \\
        &= \sumi[m]\mathcal L_{\vec \pi}^{I_m}(\vec{\pi}'^{I_{m-1}}, \pi'^{i_m}) - C\cdot \dkl{\pi_{i_m}}{\pi'_ {i_m}}
    \end{align}
\end{proof}

\subsection{Proof of Lemma \ref{lem:surrogate_gateaux}}
\label{sec:proof_lemma_3.4}
Let us recall the definition of Gâteaux derivative.
\gateaux*
\begin{lemma}
    \label{lem:state_value_gateaux}
    Let $\vec\pi$ and $\vec\pi'$ be two joint policies, then for any agent $i$ and state $s$, $V_i^{\vec\pi}(s)$ is Gâteaux differentiable and 
    \begin{equation}
        \mathfrak D_{\vec\pi'}V_i^{\vec{\pi}}(s_0) = \E_{s\sim d_{\vec\pi|s_0},\,\vec a \sim \vec{\pi}'(\cdot|s)}\left[ A_i^{\vec{\pi}}(s, \vec{a}) \right]
    \end{equation}
\end{lemma}

\begin{proof}
    Let $\vec\pi, \vec\pi'$ be two joint policies and $\eta>0$. For any agent $i$, we defined the perturbed policy $\pi^i_\eta=\pi^i+\eta(\pi'^i-\pi^i)$, which together define the perturbed joint policy $\vec\pi_\eta$.

    Let us consider the following quantity for a given state $s_0$
    \begin{equation*}
        f(\eta)=\frac{V_i^{\vec{\pi}_\eta}(s_0)-V_i^{\vec{\pi}}(s_0)}{\eta}.
    \end{equation*}

    Using Lemma \ref{lem:diff_v0}, we have 
    \begin{align*}
        f(\eta)
        &=\frac{1}{\eta}\mathbb{E}_{\tau \sim \vec{\pi}_\eta} \left[ \sum_{t=0}^{\infty} \gamma^t A_i^{\vec{\pi}}(s_t, \vec{a}_t) \mid s_0 \right] \\
        &=\frac{1}{\eta}\sum_{t=0}^{\infty} \gamma^t \mathbb{E}_{\tau \sim \vec{\pi}_\eta} \left[ A_i^{\vec{\pi}}(s_t, \vec{a}_t) \mid s_0 \right] \\
        &=\frac{1}{\eta}\sum_{t=0}^{\infty} \gamma^t \sum_s\mathbb P(s_t=s|s_0,\vec\pi_\eta)\E_{\vec a \sim \vec{\pi}_\eta(\cdot|s)} \left[ A_i^{\vec{\pi}}(s, \vec{a}) \right]
    \end{align*}
    Note that we can swap the expectation with the sum using Fubini's Theorem as the expected sum of the discounted advantage function is converging uniformly.

    By definition of $\vec\pi_\eta$, we have 
    \begin{equation*}
        \E_{\vec a \sim \vec{\pi}_\eta(\cdot|s)} \left[ A_i^{\vec{\pi}}(s, \vec{a}) \right]=(1-\eta)\E_{\vec a \sim \vec{\pi}(\cdot|s)} \left[ A_i^{\vec{\pi}}(s, \vec{a}) \right]+\eta\E_{\vec a \sim \vec{\pi}'(\cdot|s)} \left[ A_i^{\vec{\pi}}(s, \vec{a}) \right]
    \end{equation*}
    Lemma \ref{lem:expected_advantage_null} states that the first term is null. Therefore, we have
    \begin{equation*}
        \E_{\vec a \sim \vec{\pi}_\eta(\cdot|s)} \left[ A_i^{\vec{\pi}}(s, \vec{a}) \right]=\eta\E_{\vec a \sim \vec{\pi}'(\cdot|s)} \left[ A_i^{\vec{\pi}}(s, \vec{a}) \right].
    \end{equation*}

    Putting back into $f(\eta)$ gives:
    \begin{align*}
        f(\eta)
        &=\sum_{t=0}^{\infty} \gamma^t \sum_s\mathbb P(s_t=s|s_0)\E_{\vec a \sim \vec{\pi}'(\cdot|s)} \left[ A_i^{\vec{\pi}}(s, \vec{a}) \right] \\
        &= \sum_sd_{\vec\pi_\eta}(s|s_0)\E_{\vec a \sim \vec{\pi}'(\cdot|s)}\left[ A_i^{\vec{\pi}}(s, \vec{a}) \right]
    \end{align*}

    When $\eta\rightarrow0$, Lemma \ref{lem:state_dist_continuity} states that $d_{\vec\pi_\eta}$ converges to $d_{\vec\pi}$ by continuity. Hence, we have 
    \begin{align*}
        \mathfrak D_{\vec\pi'} V^{\vec\pi}_i(s_0)
        &=\lim_{\eta\rightarrow 0}f(\eta) \\
        &=\sum_sd_{\vec\pi}(s|s_0)\E_{\vec a \sim \vec{\pi}'(\cdot|s)}\left[ A_i^{\vec{\pi}}(s, \vec{a}) \right]\\
        &=\E_{s\sim d_{\vec\pi|s_0},\,\vec a \sim \vec{\pi}'(\cdot|s)}\left[ A_i^{\vec{\pi}}(s, \vec{a}) \right]
    \end{align*}
\end{proof}

\surrogategateaux*

\begin{proof}
    Let $\vec\pi, \vec\pi'$ be two joint policies and $\eta>0$. We defined the perturbed policy $\vec\pi_\eta=\vec\pi+\eta(\vec\pi'-\vec\pi)$.

    Let us consider the following quantity
    \begin{equation*}
        f(\eta)=\frac{J(\vec{\pi}_\eta)-J(\vec{\pi})}{\eta}.
    \end{equation*}

    Then, we have
    \begin{align*}
        f(\eta)
        &=\frac{J(\vec{\pi}_\eta)-J(\vec{\pi})}{\eta} \\
        &=\frac{1}{\eta}\E_{s_0\sim\rho_0}\left[\sumi U_\alpha(\nu+V_j^{\vec\pi_\eta}(s_0))- U_\alpha(\nu+V_j^{\vec\pi}(s_0))\right] \\
        &=\sum_{s_0}\rho_0(s_0)\sumi\frac{U_\alpha(\nu+V_j^{\vec\pi_\eta}(s_0))- U_\alpha(\nu+V_j^{\vec\pi}(s_0))}{\eta}
    \end{align*}

    As the state space is finite, we can pass to the limit in the sum. 
    \begin{align*}
        \lim_{\eta\rightarrow0}f(\eta)
        &=\sum_{s_0}\rho_0(s_0)\sumi\lim_{\eta\rightarrow0}\frac{U_\alpha(\nu+V_j^{\vec\pi_\eta}(s_0))- U_\alpha(\nu+V_j^{\vec\pi}(s_0))}{\eta}
    \end{align*}

    We recognize the derivative of a composite function. As $U_\alpha$ is differentiable on $\mathbb R^+$ and $V_j^{\vec\pi}(s_0)$ is Gateaux derivative, the composite $U_\alpha(\nu+V_j^{\vec\pi}(s_0))$ is also Gateaux differentiable and the following statement holds:
    \begin{equation*}
        \mathfrak D_{\vec\pi'}J(\vec\pi)=\lim_{\eta\rightarrow0}f(\eta)=\sum_{s_0}\rho_0(s_0)\sumi U_\alpha'(\nu+V_j^{\vec\pi}(s_0))  \mathfrak D_{\vec\pi'}V_j^{\vec\pi}(s_0)
    \end{equation*}

    Using Lemma \ref{lem:state_value_gateaux}, we have 
    \begin{align*}
        \mathfrak D_{\vec\pi'}J(\vec\pi)
        &=\sum_{s_0}\rho_0(s_0) \sumi U_\alpha'(\nu+V_j^{\vec\pi}(s_0))  \mathfrak D_{\vec\pi'}V_j^{\vec\pi}(s_0)\\
        &=\E_{s_0\sim\rho_0}\left[\sumi \frac{1}{\mathcal V_j^{\vec\pi}(s_0)}\E_{s\sim d_{\vec\pi|s_0},\,\vec a \sim \vec{\pi}'(\cdot|s)}\left[ A_i^{\vec{\pi}}(s, \vec{a}) \right]\right].
    \end{align*}

    Recall we are using contextualized states, thus the initial state $s_0$ is passed to any other state in a trajectory. Then, using the total expectation theorem, we have 
    \begin{equation*}
        \mathfrak D_{\vec\pi'}J(\vec\pi)
        = \E_{s\sim d_{\vec\pi},\,\vec a \sim \vec{\pi}'(\cdot|s)}\left[ \sumi \frac{1}{\mathcal V_j^{\vec\pi}(s_0)} A_i^{\vec{\pi}}(s, \vec{a}) \right]=\mathcal L_{\vec\pi}(\vec\pi').
    \end{equation*}
\end{proof}

\subsection{Proof of Theorem \ref{thm:convergence}}
\label{sec:proof_theorem}

To prove Theorem \ref{thm:convergence}, we need the following lemmas.
\begin{lemma}
    \label{lem:expected_local_advantage_null}
    For any subset $I_m$, any joint policies $\vec\pi$ and $\vec\pi'$
    \begin{equation}
        \E_{s\sim d_{\vec\pi},\, \vec{a}^{I_{m-1}}\sim\vec\pi'^{I_{m-1}},\, a^{i_m}\sim\pi^{i_m}}\left[A_F^{\vec \pi,\, i_m}(s, \vec a^{I_{m-1}}, a^{i_m}) \right] = 0
    \end{equation}
\end{lemma}

\begin{proof}
    By definition, we have
    \begin{align}
        &\E_{s\sim d_{\vec\pi},\, \vec{a}^{I_{m-1}}\sim\vec\pi'^{I_{m-1}},\, a^{i_m}\sim\pi^{i_m}}\left[A_F^{\vec \pi,\, i_m}(s, \vec a^{I_{m-1}}, a^{i_m}) \right] \\
        &\qquad= \E_{s\sim d_{\vec\pi},\, \vec{a}^{I_{m-1}}\sim\vec\pi'^{I_{m-1}},\, a^{i_m}\sim\pi^{i_m}}\left[\sumi\frac{1}{\calvf(s)} A_j^{\vec \pi,\, i_m}(s, \vec a^{I_{m-1}}, a^{i_m}) \right] \\
        &\qquad= \sumi\E_{s\sim d_{\vec\pi},\, \vec{a}^{I_{m-1}}\sim\vec\pi'^{I_{m-1}},\, a^{i_m}\sim\pi^{i_m}}\left[\frac{1}{\calvf(s)} A_j^{\vec \pi,\, i_m}(s, \vec a^{I_{m-1}}, a^{i_m}) \right] \\
        &\qquad= \sumi\E_{s\sim d_{\vec\pi}}\left[\frac{1}{\calvf(s)}\E_{\vec{a}^{I_{m-1}}\sim\vec\pi'^{I_{m-1}},\, a^{i_m}\sim\pi^{i_m} |s}\left[ A_j^{\vec \pi,\, i_m}(s, \vec a^{I_{m-1}}, a^{i_m}) \right]\right].
    \end{align}

    For any agent $j$, we have
    \begin{align}
        &\E_{\vec{a}^{I_{m-1}}\sim\vec\pi'^{I_{m-1}},\, a^{i_m}\sim\pi^{i_m} |s}\left[ A_j^{\vec \pi,\, i_m}(s, \vec a^{I_{m-1}}, a^{i_m}) \right] \\
        &\qquad= \E_{\vec{a}^{I_{m-1}}\sim\vec\pi'^{I_{m-1}},\, a^{i_m}\sim\pi^{i_m} |s}\left[ Q_j^{\vec \pi,\, I_m}(s, \, \vec{a}^{I_{m-1}}, \, a^{i_m}) - Q_j^{\vec \pi,\, I_{m-1}}(s, \, \vec{a}^{I_{m-1}}) \right] \\
        &\qquad= \E_{\vec{a}^{I_{m-1}}\sim\vec\pi'^{I_{m-1}},\, a^{i_m}\sim\pi^{i_m} |s}\left[ Q_j^{\vec \pi,\, I_m}(s, \, \vec{a}^{I_{m-1}}, \, a^{i_m})\right]\notag\\ 
        &\qquad\qquad\qquad\qquad\qquad- \E_{\vec{a}^{I_{m-1}}\sim\vec\pi'^{I_{m-1}}|s}\left[Q_j^{\vec \pi,\, I_{m-1}}(s, \, \vec{a}^{I_{m-1}}) \right]
    \end{align}

    The first term can be re-written as 
    \begin{align}
        &\E_{\vec{a}^{I_{m-1}}\sim\vec\pi'^{I_{m-1}},\, a^{i_m}\sim\pi^{i_m} |s}\left[ Q_j^{\vec \pi,\, I_m}(s, \, \vec{a}^{I_{m-1}}, \, a^{i_m})\right]\\ 
        &\quad= \E_{\vec{a}^{I_{m-1}}\sim\vec\pi'^{I_{m-1}},\, a^{i_m}\sim\pi^{i_m} |s}\left[\E_{\vec a^{-I_m}\sim \vec \pi_{-I_m}}\left[Q_j^{\vec \pi}(s,\, \vec{a}^{I_{m-1}}, a^{i_m},\, \vec{a}^{-I_m})\right]\right] \\
        &\quad= \E_{\vec{a}^{I_{m-1}}\sim\vec\pi'^{I_{m-1}}|s}\left[\E_{\vec a^{-I_{m-1}}\sim \vec \pi^{-I_{m-1}}}\left[Q_j^{\vec \pi}(s,\, \vec{a}^{I_{m-1}},\, \vec{a}^{-I_{m-1}})\right]\right] \\
        &\quad= \E_{\vec{a}^{I_{m-1}}\sim\vec\pi'^{I_{m-1}}|s}\left[Q_j^{\vec\pi,\, I_{m-1}}(s, \vec a^{I_{m-1}})\right]
    \end{align}

    Putting back into the former equation ends the demonstration:
    \begin{align}
        &\E_{\vec{a}^{I_{m-1}}\sim\vec\pi'^{I_{m-1}},\, a^{i_m}\sim\pi^{i_m} |s}\left[ A_j^{\vec \pi,\, i_m}(s, \vec a^{I_{m-1}}, a^{i_m}) \right] \\
        &\quad= \E_{\vec{a}^{I_{m-1}}\sim\vec\pi'^{I_{m-1}},\, a^{i_m}\sim\pi^{i_m} |s}\left[ Q_j^{\vec \pi,\, I_m}(s, \, \vec{a}^{I_{m-1}}, \, a^{i_m})\right]\notag\\
        &\qquad\qquad\qquad\qquad- \E_{\vec{a}^{I_{m-1}}\sim\vec\pi'^{I_{m-1}}|s}\left[Q_j^{\vec \pi,\, I_{m-1}}(s, \, \vec{a}^{I_{m-1}}) \right] \\
        &\quad= \E_{\vec{a}^{I_{m-1}}\sim\vec\pi'^{I_{m-1}}|s}\left[Q_j^{\vec \pi,\, I_{m-1}}(s, \, \vec{a}^{I_{m-1}}) \right] - \E_{\vec{a}^{I_{m-1}}\sim\vec\pi'^{I_{m-1}}|s}\left[Q_j^{\vec \pi,\, I_{m-1}}(s, \, \vec{a}^{I_{m-1}}) \right] \\
        &\quad =0
    \end{align}
\end{proof}

\convergence*

\begin{proof}
    We will prove the theorem in three folds.
    
    \textbf{I. Monotonic Improvement Property}
    
    For any $k\in \mathbb N$ and any permutation $I_n$, suppose $\vec \pi_{k+1}$ is updated by Algorithm \ref{algo:fhatrl}, i.e. 
    \begin{equation}
        \pi_{k+1}^{i_m}=\underset{\pi^{i_m}}{\arg \max}\;\mathcal L_{\vec \pi^i}^{I_m}(\vec{\pi}_{k+1}^{I_{m-1}}, \pi^{i_m}) - C\cdot \dkl{\pi_k^{i_m}}{\pi^ {i_m}}
    \end{equation}
   
    We want to bound the difference $J(\vec \pi_{k+1}) -  J(\vec\pi^i)$:
    \begin{align*}
        J(\vec \pi_{k+1}) -  J(\vec\pi^i)
        &\geq \sumi[l]\mathcal L_{\vec \pi^i}^{I_m}(\vec{\pi}_{k+1}^{I_{m-1}}, \pi_{k+1}^{i_m}) - C\cdot \dkl{\pi_{k}^{i_m}}{\pi_{k+1}^{i_m}} \\
    \end{align*}

     As $\pi_{k+1}^{i_m}$ maximizes the surrogate function, the latter function gets lower evaluation for any other policy $\pi'^{i_m}$. In particular, for $\pi'^{i_m}=\pi_{k}^{i_m}$, we have :
    \begin{align}
        \mathcal L_{\vec \pi^i}^{I_m}(\vec{\pi}_{k+1}^{I_{m-1}}, \pi_{k+1}^{i_m}) - C\cdot \dkl{\pi_k^{i_m}}{\pi_{k+1}^{i_m}}
        &\geq \mathcal L_{\vec \pi^i}^{I_m}(\vec{\pi}_{k+1}^{I_{m-1}}, \pi_{k}^{i_m}) - C\cdot \dkl{\pi_k^{i_m}}{\pi_{k}^{i_m}} \\
        &\geq \mathcal L_{\vec \pi^i}^{I_m}(\vec{\pi}_{k+1}^{I_{m-1}}, \pi_{k}^{i_m})
    \end{align}

    We then use the Lemma \ref{lem:expected_local_advantage_null}:
    \begin{align}
        \mathcal L_{\vec \pi^i}^{I_m}(\vec{\pi}_{k+1}^{I_{m-1}}, \pi_{k}^{i_m})
        &=\E_{s_\sim d_{\vec\pi^i}, \;\vec{a}^{I_{m-1}}\sim\vec{\pi}_{k+1}^{I_{m-1}}(.|s), \, a^{i_m}\sim\pi_k^{i_m}(.|s)}\left[A^{\vec\pi,\,i_m}\left(s, \vec a ^{I_{m-1}}, a^{i_m}\right)\right] \\
        &=0.
    \end{align}

    Hence, it follows that 
    \begin{align*}
        J(\vec \pi_{k+1}) -  J(\vec\pi^i)
        &\geq \sumi[l]\mathcal L_{\vec \pi^i}^{I_m}(\vec{\pi}_{k+1}^{I_{m-1}}, \pi_{k+1}^{i_m}) - C\cdot \dkl{\pi_{k}^{i_m}}{\pi_{k+1}^{i_m}} \\
        &\geq \sumi[l]\mathcal L_{\vec \pi^i}^{I_m}(\vec{\pi}_{k+1}^{I_{m-1}}, \pi_{k}^{i_m})\\
        &\geq 0.
    \end{align*}

    \textbf{II. Convergence of the fair global objective}
    
    First of all, the sequence $(J(\vec\pi_k))_{k\in\mathbb N}$ converges as, by the Monotonic Improvement property, it is monotonically increasing and upper bounded by $n\cdot U_\alpha\left(\nu+\frac{r_{\max}}{1-\gamma}\right)$. Let us denote the limit by $J_\infty$.

    \textbf{III. Nash Equilibrium of the limit points}
    
    The following proof is adapted from Appendix C.3 of \cite{hatrl}. We update the original notation and substitute their multi-agent advantage function with our fair advantage function. While Steps 1 and 2 follow directly from this substitution, we include the fully adapted steps here for completeness.
    
    \textbf{Step 1 (Stationarity of any limit point).} The policy state $\Pi^i$ is bounded. Therefore, for any sequence of policies $(\vec \pi^i)$, Bolzano-Weierstrass Theeorem states that we can extract a subsequence $(\vec\pi_{k_j})_{j\in\mathbb N}$ that converges to an adherent point $\vec{\pi}_\infty$.
    By continuity of $J$ in $\vec\pi$ (\ref{cor:value_function_continuity}), we have
    \begin{equation}
        J(\vec{\pi}_\infty) = J\left( \lim_{j \to \infty} \vec\pi_{k_j} \right)= \lim_{j \to \infty} J(\vec\pi_{k_j}) = J_\infty.
    \end{equation}

    \begin{definition}[TR-Stationarity, \cite{hatrl}]
        \label{def:tr_stationarity}
        A joint policy $\vec\pi_\circ$ is a trust-region-stationary (TR-stationary) if, for every agent $i$,
        \begin{equation}
            \pi^i_\circ = \arg \max_{\pi^i} \left[ \mathbb{E}_{s \sim d_{\vec{\pi}_\circ},\, a^i \sim \pi^i} \left[ A_F^{\vec{\pi}_\circ,\, i}(s, a^i) \right] - C_{\vec{\pi}_\circ} \mathrm{D}_{KL}^{\max} \left( \pi^i_\circ, \pi^i \right) \right],
        \end{equation}
        where $C_{\vec{\pi}_\circ} = \frac{4nL_\alpha}{(1-\gamma)^2}A_{\max}^2 + \frac{4nA_{\max}\gamma}{\nu^\alpha(1-\gamma)^2 }$, and $A_{\max} = \max_{j, s, \vec{a}} |A_j^{\vec{\pi}_\circ}(s, \vec{a})|$.
    \end{definition}

    The goal now is to establish the TR-stationarity of any limit point joint policy $\vec{\pi}_\infty$. Let $\mathbb{E}_{I_{n}^{0:\infty}} [\cdot]$ denote the expected value operator under the random process $(I_{n}^{0:\infty})$. Let also $A_k = \max_{j, \,s, \,\vec{a}} |A_j^{\vec{\pi}_k}(s, \vec{a})|$, and $C_{\vec{\pi}_k} = \frac{4nL_\alpha}{(1-\gamma)^2}A_{k}^2 + \frac{4nA_{k}\gamma}{\nu^\alpha(1-\gamma)^2 }$. We have
    \begin{align*}
        0 &= \lim_{k \to \infty} \mathbb{E}_{I_{n}^{0:\infty}} [J(\vec{\pi}_{k+1}) - J(\vec{\pi}_k)] \\
        &\ge \lim_{k \to \infty} \mathbb{E}_{I_{n}^{0:\infty}} [\mathcal L_{\vec{\pi}_k}(\vec{\pi}_{k+1}) - C_k \mathrm{D}_{KL}^{\max}(\vec{\pi}_k, \vec{\pi}_{k+1})] \text{ by Lemma \textcolor{red}{\ref{lem:fair_surrogate}}} \\
        &\ge \lim_{k \to \infty} \mathbb{E}_{I_{n}^{0:\infty}} \left[ \mathcal L_{\vec{\pi}_k}^{i^k_1} \left( \pi_{k+1}^{i^k_1} \right) - C_k \mathrm{D}_{KL}^{\max} \left( \pi_k^{i^k_1}, \pi_{k+1}^{i^k_1} \right) \right],
    \end{align*}
    by Lemma \textcolor{red}{\ref{lem:surrogate_decomp}} and removing the all the summands except $i_1^k$ (because the summands are all positive). 
    
    Now, we consider an arbitrary limit point $\vec\pi_\infty$ from the adherent set, and a subsequence $(\vec\pi_{k_j})_{j=0}^{\infty}$ that converges to $\vec{\pi}_\infty$. 
    
    On one hand, we have from above
    
    \begin{equation}
        0 \ge \lim_{j \to \infty} \mathbb{E}_{I_{n}^{0:\infty}} \left[
    \mathcal L^{i^{k_j}_1}_{\vec\pi_{k_j}} \left( \pi^{i^{k_j}_1}_{k_j+1} \right)
    - C_{k_j} D_{\mathrm{KL}}^{\max} \left( \pi^{i^{k_j}_1}_{k_j}, \pi^{i^{k_j}_1}_{k_j+1} \right)
    \right].
    \end{equation}
    
    On the other hand, let denote by $p_i$ the probability to have $i_{k_j} = i$ under the assumption that every permutation has strictly positive probability to happen. As the expectation is taken of positive random variables, we also have

    \begin{align*}
        &\lim_{j \to \infty} \mathbb{E}_{I_{n}^{0:\infty}} \left[\mathcal L^{i^{k_j}_1}_{\vec\pi_{k_j}} \left( \pi^{i^{k_j}_1}_{k_j+1} \right)- C_{k_j} D_{\mathrm{KL}}^{\max} \left( \pi^{i^{k_j}_1}_{k_j}, \pi^{i^{k_j}_1}_{k_j+1} \right)\right]\\
        & \qquad\qquad\qquad\qquad\qquad\qquad\qquad\geq 
        p_i \lim_{j \to \infty} \max_{\pi^i} \left[\mathcal L^{i}_{\vec\pi_{k_j}} \left( \pi^i \right) - C_{k_j} D_{\mathrm{KL}}^{\max} \left( \pi^i_{k_j}, \pi^i \right)\right] \\
        & \qquad\qquad\qquad\qquad\qquad\qquad\qquad\geq 
        p_i \max_{\pi^i} \left[\mathcal L^{i}_{\vec{\pi}_\infty} \left( \pi^i \right) - C_{\vec{\pi}_\infty} D_{\mathrm{KL}}^{\max} \left( \vec{\pi}^i_\infty, \pi^i \right)\right] \\
        & \qquad\qquad\qquad\qquad\qquad\qquad\qquad\geq 0 \text{ by Lemma \ref{lem:expected_local_advantage_null}.}
    \end{align*}

    $\mathcal L^{i}_{\vec\pi_{k_j}}(\pi^i)$ converges as $\mathcal L^{i}_{\vec\pi}(\pi^i)$ is continuous with respect to $\vec\pi$ (by Definition \ref{def:local_surrogate} and Lemma \ref{lem:continuity_expected_fair_adv}). $C_{k_j}$ converges as it is continuous with respect to $\vec\pi$ (Corollary \ref{cor:value_function_continuity}). $D_{\mathrm{KL}}^{\max}$ convergence follows from from continuity of $D_{\mathrm{KL}}$ with Assumption \ref{assumption} and continuity of $\max$ over finite state space.

    Therefore, we have 
    \begin{equation}
        \max_{\pi^i} \left[\mathcal L^{i}_{\vec{\pi}_\infty}(\pi^i)- C_{\vec{\pi}_\infty} D_{\mathrm{KL}}^{\max} \left( \pi^i_\infty, \pi^i \right)\right] = 0,
    \end{equation}
    which proves $\vec\pi_\infty$ is TR-stationary. 
    
    \textbf{Step 2 (dropping the penalty term).} In this second step, the goal is to prove that a TR-stationary joint policy $\bar{\pi}$ satisfies 
    \begin{equation}
        \label{eq:arg_max}
        \pi_\infty^i = \arg\max_{\pi^i} \mathbb{E}_{a^i \sim \pi^i} \left[ A_F^{\vec{\pi}_\infty, \, i} (s, a^i) \right],
    \end{equation}
    for every state $s \in \mathcal S$.
    
    We will prove the statement by contradiction. Hence, suppose that there exists a state $s_0$ and a policy $\hat{\pi}^i$ such that
    \begin{equation}
        \label{eq:contradiction_hypothesis}
        \mathbb{E}_{a^i \sim \hat{\pi}^i} \left[ A_F^{\vec{\pi}_\infty,\, i} (s_0, a^i) \right] > \mathbb{E}_{a^i \sim \pi^i_\infty} \left[ A_F^{\vec{\pi}_\infty,\, i} (s_0, a^i) \right].
    \end{equation}

    Parameterize the policy $\pi^i(\cdot \mid s_0)$ as a probability vector:
    \begin{equation}
        \pi^i(\cdot \mid s_0) = \left( x_1^i, \ldots, x_{d^i-1}^i, 1 - \sum_{j=1}^{d^i-1} x_j^i \right),
    \end{equation}
    where $x_j^i$ ensures a valid distribution and $d^i=|\mathcal A_i|$. The expected fair advantage in Equation \ref{eq:arg_max} can be rewritten as:
    \begin{align}
        \mathbb{E}_{a^i \sim \pi^i} \left[A_F^{\vec{\pi}_\infty,\, i}(s_0, a^i) \right] 
        &= \sum_{j=1}^{d^i-1} x_j^i \cdot A_F^{\vec{\pi}_\infty,\, i} (s_0, a_j^i) + \left( 1 - \sum_{h=1}^{d^i-1} x_h^i \right) A_F^{\vec{\pi}_\infty,\, i} (s_0, a_{d^i}^i) \\
        &= \sum_{j=1}^{d^i-1} x_j^i \left[ A_F^{\vec{\pi}_\infty,\, i}(s_0, a_j^i) - A_F^{\vec{\pi}_\infty,\, i} (s_0, a_{d^i}^i) \right] + A_F^{\vec{\pi}_\infty,\, i}(s_0, a_{d^i}^i).
    \end{align}
    
    Because this expectation is an affine function of $x^i$, its gradients and directional derivatives are constant. The existence of $\hat{\pi}^i(\cdot \mid s_0)$ satisfying Inequality \ref{eq:contradiction_hypothesis} guarantees a strictly positive directional derivative from $\bar{\pi}^i(\cdot \mid s_0)$ toward $\hat{\pi}^i(\cdot \mid s_0)$. 
    
    Additionally, the KL divergence gradient vanishes at the reference policy $\pi_\infty^i$:
    \begin{align}
        \frac{\partial D_{\mathrm{KL}}(\pi_\infty^i(\cdot \mid s_0), \pi^i(\cdot \mid s_0))}{\partial x_j^i} 
        &= \frac{\partial}{\partial x_j^i} \left[ (\pi_\infty^i(\cdot \mid s_0))^T \big( \log \pi_\infty^i(\cdot \mid s_0) - \log \pi^i(\cdot \mid s_0) \big) \right] \\
        &= - \frac{\partial}{\partial x_j^i} \left[ (\pi_\infty^i)^T \log \pi^i \right] \quad \text{(omitting state $s_0$ for brevity)} \\
        &= - \frac{\partial}{\partial x_j^i} \sum_{k=1}^{d^i-1} \pi_{\infty, k}^i \log x_k^i \;-\; \frac{\partial}{\partial x_j^i} \pi_{\infty,\,d^i}^i \log \left( 1 - \sum_{k=1}^{d^i-1} x_k^i \right) \\
        &= - \frac{\pi_{\infty,\,j}^i}{x_j^i} + \frac{\pi_{\infty,\,d^i}^i}{1 - \sum_{k=1}^{d^i-1} x_k^i} \\
        &= - \frac{\pi_{\infty,\,j}^i}{\pi_j^i} + \frac{\pi_{\infty,\,d^i}^i}{\pi_{d^i}^i} = 0, \quad \text{when } \pi^i = \pi_\infty^i.
    \end{align}
    
    Consequently, when evaluated at $\pi^i(\cdot \mid s_0) = \pi_\infty^i(\cdot \mid s_0)$, the full penalized objective:
    \begin{equation}
        d_{\vec\pi_\infty}(s_0)\,\mathbb{E}_{a^i \sim \pi^i} \left[ A_F^{\vec{\pi}_\infty,\, i} (s_0, a^i) \right] - C_{\vec{\pi}_\infty} D_{\mathrm{KL}} \big( \pi_\infty^i(\cdot \mid s_0), \pi^i(\cdot \mid s_0) \big)
    \end{equation}
    has a strictly positive directional derivative toward $\hat{\pi}^i(\cdot \mid s_0)$. This implies we can find a local policy $\tilde{\pi}^i(\cdot \mid s_0)$ on the path toward $\hat{\pi}^i(\cdot \mid s_0)$ that strictly improves the objective:
    \begin{equation}
    d_{\vec\pi_\infty}(s_0)\,\mathbb{E}_{a^i \sim \tilde{\pi}^i} \left[ A_F^{\vec{\pi}_\infty,\, i} (s_0, a^i) \right] 
    - C_{\vec{\pi}_\infty} D_{\mathrm{KL}} \big( \pi_\infty^i(\cdot \mid s_0), \tilde{\pi}^i(\cdot \mid s_0) \big) 
    > 0.
    \end{equation}
    
    Now, construct a global policy $\pi_*^i$ such that $\pi_*^i(\cdot \mid s_0) = \tilde{\pi}^i(\cdot \mid s_0)$ and $\pi_*^i(\cdot \mid s) = \pi_\infty^i(\cdot \mid s)$ for all $s \neq s_0$. For unperturbed states ($s \neq s_0$), both the advantage and KL divergence evaluate to zero:
    \begin{align}
        d_{\vec{\pi}_\infty}(s)\,\mathbb{E}_{a^i \sim \pi_*^i} \left[ A_F^{\vec{\pi}_\infty,\, i} (s, a^i) \right] 
        &= d_{\vec{\pi}_\infty}(s)\,\mathbb{E}_{a^i \sim \pi_\infty^i} \left[ A_F^{\vec{\pi}_\infty,\, i}(s, a^i) \right] = 0, \notag\\
        D_{\mathrm{KL}} \big( \pi_\infty^i(\cdot \mid s), \pi_*^i(\cdot \mid s) \big) &= 0.
    \end{align}
    
    Applying this to the full objective yields:
    \begin{align}
        \mathcal L^i_{\vec{\pi}_\infty}(\pi_*^i) - C_{\vec{\pi}_\infty} D_{\mathrm{KL}}^{\max}(\pi_\infty^{i}, \pi_*^i)
        &= d_{\vec{\pi}_\infty}(s_0)\,\mathbb{E}_{a^i \sim \tilde{\pi}^i} \left[ A_F^{\vec{\pi}_\infty,\, i} (s_0, a^i) \right] \notag\\
        &\qquad - C_{\vec{\pi}_\infty} D_{\mathrm{KL}} \big( \pi_\infty^i(\cdot \mid s_0), \tilde{\pi}^i(\cdot \mid s_0) \big) \\
        &> 0 \\
        &= \mathcal L^i_{\vec{\pi}_\infty}(\pi_\infty^i) - C_{\vec{\pi}_\infty} D_{\mathrm{KL}}^{\max}(\pi_\infty^{i}, \pi_\infty^i).
    \end{align}
    
    This strict improvement directly contradicts the assumption of TR-stationarity for $\bar{\pi}$, thereby proving the claim.

    \textbf{Step 3 (optimality).} Now, for a fixed joint policy $\vec{\pi}_\infty^{-i}$ of other agents, $\pi_\infty^i$ satisfies 
    \begin{equation}
        \pi_\infty^i = \arg\max_{\pi^i} \mathbb{E}_{a^i \sim \pi^i} \left[ A^{\vec{\pi}_\infty, \, i} (s, a^i) \right].
    \end{equation}

    Lemma \ref{lem:fair_surrogate} gives an upper bound of the performance difference. Then, for a given TR-stationary point $\vec\pi_\infty$ and any other policy $\pi^i$ for agent $i$, we have  
    \begin{equation*}
        J(\pi^i,\vec\pi_\infty^{-i})-J(\vec\pi_\infty) \leq \mathcal L _{\vec\pi_\infty}(\pi^i,\vec\pi_\infty^{-i}),
    \end{equation*}
    which by definition of $\mathcal L_{\vec\pi_\infty}$ is equivalent to 
    \begin{align}
        J(\pi^i,\vec\pi_\infty^{-i})-J(\vec\pi_\infty) 
        &\leq \mathbb{E}_{s\sim d_{\vec\pi_\infty}, \; a^i\sim\pi^i(\cdot|s),\, \vec{a}^{-i}\sim\vec{\pi}^{-i}_\infty(\cdot|s)}\left[A_{F}^{\vec\pi_\infty}(s,a^i,\vec a^{-i}) \right] \\
        &=\mathbb{E}_{s\sim d_{\vec\pi_\infty}}\left[\E_{a^i\sim\pi^i(\cdot|s)}\left[A_{F}^{\vec\pi_\infty, \,i}(s,a^i) \right] \right]\\
        &\leq \mathbb{E}_{s\sim d_{\vec\pi_\infty}}\left[\E_{a^i\sim\pi^i_\infty(\cdot|s)}\left[A_{F}^{\vec\pi_\infty, \,i}(s,a^i) \right] \right] \text{ (by definition  of $\pi_\infty^i$)} \\
        &=0\qquad\qquad\qquad\qquad\qquad\qquad\qquad\quad\; \text{ (by Lemma \ref{lem:expected_local_advantage_null})}
    \end{align}

    Thus, we have 
    \begin{equation}
        J(\pi^i,\vec\pi_\infty^{-i})\leq J(\vec\pi_\infty).
    \end{equation}

    As this statement holds for any arbitrary agent $i$, this proves $\vec\pi_\infty$ is a Nash equilibrium.
\end{proof}

\section{Pratical Algorithms}
\subsection{$\alpha$-fair HATRPO}
The derivation of $\alpha$-fair HATRPO relies on an estimator $\hat{g}_{k}^{i_{m}}$ of the gradient of the local objective $\E_{s_\sim d_{\vec\theta_k}, \;\vec{a}^{I_{m-1}}\sim\vec{\theta}_{k+1}^{I_{m-1}}(.|s), \, a^{i_m}\sim\theta^{i_m}(.|s)}\left[A_F^{\vec\theta_k,\,i_m}\left(s, \vec a ^{I_{m-1}}, a^{i_m}\right)\right]$.  

In this form, the local objective is hard to differentiate. Thankfully, we can re-write this in a more tractable expression.
\begin{lemma}
    Let $\vec\pi$ be a joint policy, $\vec\pi'^{I_{m-1}}$ be some other joint policy of agents $I_{m-1}$ and $\pi'^{i_m}$ be some other policy of agent $i_m$. Then, the local objective can be rewritten as 
    \begin{align}
        \mathbb{E}_{\vec{a}^{I_{m-1}} \sim \vec{\pi}'^{I_{m-1}}, a^{i_m} \sim \pi'^{i_m}} [A_F^{\vec{\pi},\,i_m}(s,& \vec{a}^{I_{m-1}}, a^{i_m})] = \notag \\
        &\mathbb{E}_{\vec{a} \sim \vec{\pi}} \left[\left(\frac{\pi'^{i_m}({a}^{i_m}|s)}{\pi^{i_m}({a}^{i_m}|s)} - 1\right)\frac{\vec{\pi}'^{I_{m-1}}(\vec{a}^{I_{m-1}}|s)}{\vec\pi^{I_{m-1}}(\vec{a}^{I_{m-1}}|s)}A_F^{\vec{\pi}}(s, \vec{a})\right].
    \end{align}
\end{lemma}

\begin{proof}
    See \cite{hatrl}, Appendix D.1 (Proposition 2). The proof is similar by replacing the multi-agent advantage function with our fair advantage function. Because both advantage functions enjoy the same decomposition property (Lemma \ref{lem:advantage_decomp}), the proof holds for our fair advantage function.
\end{proof}

Let us introduce $M_F^{I_m} = \frac{\vec{\pi}'^{I_{m-1}}(\vec{a}^{I_{m-1}} | s)}{\vec{\pi}^{I_{m-1}}(\vec{a}^{I_{m-1}} | s)} A_F^{\vec\theta_k}(s, \vec{a})$.
It follows that the derivation of the gradient estimator for HATRPO holds for $\alpha$-fair HATRPO by replacing the multi-agent advantage function with our fair advantage function:

\begin{align}
&\nabla_{\theta^{i_m}} \mathbb{E}_{s \sim d_{\vec\theta_k}, \vec{a} \sim \vec\theta_k} \left[ \left( \frac{\pi_{\theta^{i_m}}^{i_m}(a^{i_m}|s)}{\pi_{\theta_k^{i_m}}^{i_m}(a^{i_m}|s)} - 1 \right) M_F^{I_{m}}(s, \vec{a}) \right] \\
&\qquad\qquad= \mathbb{E}_{s \sim d_{\vec\theta_k}, a \sim \vec\theta_k} \left[ \frac{\pi_{\theta^{i_m}}^{i_m}(a^{i_m}|s)}{\pi_{\theta_k^{i_m}}^{i_m}(a^{i_m}|s)} \nabla_{\theta^{i_m}} \log \pi_{\theta^{i_m}}^{i_m}(a^{i_m}|s) M_F^{i_{1:m}}(s, \vec{a}) \right].
\end{align}
See Appendix D.2 of \cite{hatrl} for more details.

\begin{algorithm}[htbp]
\caption{$\alpha$-FHATRPO}
\label{alg:hatrpo}
\begin{algorithmic}[1]
    \State \textbf{Input:} Stepsize $\omega$, batch size $B$, number of: agents $n$, episodes $K$, steps per episode $T$, constant $\nu$, fairness level $\alpha$, possible steps in line search $L$, line search acceptance threshold $\iota$.
    \State \textbf{Initialize:} Actor networks $\{\theta_0^i, \ \forall i \in \mathcal{N}\}$, individual V-value networks $\{\phi^i_0,  \ \forall i \in \mathcal{N}\}$, Replay buffer $\mathcal{B}$
    \For{$k = 0, 1, \dots, K - 1$}
        \State Collect a set of trajectories by running the joint policy $\boldsymbol{\pi}_{\boldsymbol{\theta}_k} = (\pi^1_{\theta^1_k}, \dots, \pi^n_{\theta^n_k})$.
        \State Sample a random minibatch of $B$ transitions from $\mathcal{B}$.
        \State Compute individual state value function $\hat{V}_j(s)$ based on individual critics network.
        \State Compute individual advantage function $\hat{A}_j(s, \vec{a})$ with GAE.
        \State Compute fair advantage function $\hat{A}_F(s, \vec{a})$ with $\alpha$, $\nu$, $\hat A_j(s,\vec a)$ and $\hat V_j(s)$.
        \State Draw a random permutation of agents $I_{n}$.
        \State Set $M_F^{i_1}(\mathbf{s}, \vec{a}) = \hat{A}_F(\mathbf{s}, \vec{a})$.
        \For{\textbf{agent} $i_m = i_1, \dots, i_n$}
            \State Estimate the gradient of the agent's maximisation objective
            \[
                \hat{\boldsymbol{g}}_{k}^{i_m} = \frac{1}{B} \sum_{b=1}^{B} \sum_{t=1}^{T} \nabla_{\theta_{k}^{i_m}} \log \pi_{\theta_{k}^{i_m}}^{i_m} (a_t^{i_m} \mid o_t^{i_m}) M_F^{i_{1:m}}(s_t, \boldsymbol{a}_t).
            \]
            \Statex \qquad Use the conjugate gradient algorithm to compute the update direction
            \[
                \hat{\boldsymbol{x}}_{k}^{i_m} \approx (\hat{\boldsymbol{H}}_{k}^{i_m})^{-1} \hat{\boldsymbol{g}}_{k}^{i_m},
            \]
            \Statex \qquad where $\hat{\boldsymbol{H}}_{k}^{i_m}$ is the Hessian of the average KL-divergence
            \[
                \frac{1}{BT} \sum_{b=1}^{B} \sum_{t=1}^{T} \text{D}_{\text{KL}} \left( \pi_{\theta_k^{i_m}}^{i_m} (\cdot \mid o_t^{i_m}), \pi_{\theta^{i_m}}^{i_m} (\cdot \mid o_t^{i_m}) \right).
            \]
            \State Estimate the maximal step size allowing for meeting the KL-constraint
            \[
                \hat{\beta}_{k}^{i_m} \approx \sqrt{\frac{2\delta}{(\hat{\boldsymbol{x}}_{k}^{i_m})^T \hat{\boldsymbol{H}}_{k}^{i_m} \hat{\boldsymbol{x}}_{k}^{i_m}}}.
            \]
            \State Update agent $i_m$'s policy by
            \[
                \theta_{k+1}^{i_m} = \theta_{k}^{i_m} + \omega^j \hat{\beta}_{k}^{i_m} \hat{\boldsymbol{x}}_{k}^{i_m},
            \]
            \Statex \qquad where $j \in \{0, 1, \dots, L\}$ is the smallest such $j$ which improves the sample loss by at least 
            \Statex \qquad $\iota \alpha^j \hat{\beta}_{k}^{i_m} \hat{\boldsymbol{x}}_{k}^{i_m} \cdot \hat{\boldsymbol{g}}_{k}^{i_m}$, found by the backtracking line search.
            \State Compute $M_F^{i_{1:m+1}}(\mathbf{s}, \vec{a}) = \frac{\pi_{\theta_{k+1}^{i_m}}^{i_m}(a^{i_m} \mid o^{i_m})}{\pi_{\theta_{k}^{i_m}}^{i_m}(a^{i_m} \mid o^{i_m})} M_F^{i_{1:m}}(s_t, \boldsymbol{a}_t)$. \textcolor{teal}{//Unless $m = n$.}
        \EndFor
        \State Update V-value networks by following formula:
        \[
            \phi^i_{k+1} = \arg\min_{\phi^i} \frac{1}{BT} \sum_{b=1}^{B} \sum_{t=0}^{T} \left( V_{\phi^i}(s_t) - \hat{R}_t^i \right)^2, \, \forall i\in \mathcal N.
        \]
    \EndFor
\end{algorithmic}
\end{algorithm}

\subsection{$\alpha$-fair HAPPO}
\begin{algorithm}[H]
\caption{$\alpha$-FHAPPO}
\label{alg:happo}
\begin{algorithmic}[1]
    \State \textbf{Input:} Clipping $\epsilon$, batch size $B$, number of: agents $n$, episodes $K$, steps per episode $T$.
    \State \textbf{Initialize:} Actor networks $\{\theta_0^i, \ \forall i \in \mathcal{N}\}$, individual V-value networks $\{\phi^i_0,  \ \forall i \in \mathcal{N}\}$, Replay buffer $\mathcal{B}$
    \For{$k = 0, 1, \dots, K - 1$}
        \State Collect a set of trajectories by running the joint policy $\boldsymbol{\pi}_{\boldsymbol{\theta}_k} = (\pi^1_{\theta^1_k}, \dots, \pi^n_{\theta^n_k})$.
        \State Push transitions $\{(o^i_t, a^i_t, o^i_{t+1}, r_t), \forall i \in \mathcal{N}, t \in T\}$ into $\mathcal{B}$.
        \State Sample a random minibatch of $B$ transitions from $\mathcal{B}$.
        \State Compute individual state value function $\hat{V}_j(s)$ based on individual critics network.
        \State Compute individual advantage function $\hat{A}_j(s, \vec{a})$ with GAE.
        \State Compute fair advantage function $\hat{A}_F(s, \vec{a})$ with $\alpha$, $\nu$, $\hat A_j(s,\vec a)$ and $\hat V_j(s)$.
        \State Draw a random permutation of agents $i_{1:n}$.
        \State Set $M_F^{i_1}(\mathbf{s}, \vec{a}) = \hat{A}_F(\mathbf{s}, \vec{a})$.
        \For{\textbf{agent} $i_m = i_1, \dots, i_n$}
            \State Update actor $i_m$ with $\theta_{k+1}^{i_m}$, the argmax of the PPO-Clip objective
            \[
                \frac{1}{BT} \sum_{b=1}^{B} \sum_{t=0}^{T} \min \left( \frac{\pi_{\theta^{i_m}}^{i_m} (a_t^{i_m} \mid o_t^{i_m})}{\pi_{\theta_k^{i_m}}^{i_m} (a_t^{i_m} \mid o_t^{i_m})} M_F^{i_{1:m}}(s_t, \boldsymbol{a}_t), \text{clip} \left( \frac{\pi_{\theta^{i_m}}^{i_m} (a_t^{i_m} \mid o_t^{i_m})}{\pi_{\theta_k^{i_m}}^{i_m} (a_t^{i_m} \mid o_t^{i_m})}, 1 \pm \epsilon \right) M_F^{i_{1:m}}(s_t, \boldsymbol{a}_t) \right).
            \]
            \State Compute $M_F^{i_{1:m+1}}(\mathbf{s}, \vec{a}) = \frac{\pi_{\theta_{k+1}^{i_m}}^{i_m}(a^{i_m} \mid o^{i_m})}{\pi_{\theta_{k}^{i_m}}^{i_m}(a^{i_m} \mid o^{i_m})} M_F^{i_{1:m}}(\mathbf{s}, \vec{a})$. \textcolor{teal}{//Unless $m = n$.}
        \EndFor
        \State Update V-value networks by following formula:
        \[
            \phi_{k+1}^i = \arg\min_{\phi^i} \frac{1}{BT} \sum_{b=1}^{B} \sum_{t=0}^{T} \left( V_{\phi^i}(s_t) - \hat{R}_t^i \right)^2
        \]
    \EndFor
\end{algorithmic}
\end{algorithm}

\section{Experiments}
\subsection{Environments' parameters}

In this work, we based our implementation on the environment provided by \href{https://github.com/cooperativex/SocialJax/tree/main}{SocialJax}. 

\begin{table}[H]
\centering
\caption{CleanUp Environment Parameters}
\label{tab:cleanup_params}
\begin{tabular}{ll}
\toprule
\textbf{Parameter} & \textbf{Value / Description} \\
\midrule
Grid Size & $28 \times 19$ \\
Trajectory Length & $500$ steps \\
Max Regrowth Apple Rate & 0.05 \\
Maximum Level of Pollution & $40\%$ \\
Delay Start of Dirt Spawning & 50 \\
Dirt Spawn Probability & 0.5 \\
\bottomrule
\end{tabular}
\end{table}

\begin{table}[H]
\centering
\caption{Harvest Environment Parameters}
\label{tab:harvest_params}
\begin{tabular}{ll}
\toprule
\textbf{Parameter} & \textbf{Value} \\
\midrule
Grid Size & $22 \times 16$ \\
Trajectory Length & $500$\\
Regrowth Apple Rate & Proportional to density \\
\bottomrule
\end{tabular}
\end{table}

\subsection{Hyper-parameters of the different algorithms}
We implement HATRPO and HAPPO from scratch using the pseudo code provided by HATRL \cite{hatrl}. We recover the code of Fair MAPPO from FairMARL\cite{moi} (\href{https://github.com/AkuBrains/altruistic-fair-MARL/}{https://github.com/AkuBrains/altruistic-fair-MARL/}).
In this work, actor and critic networks use similar architecture : CNN followed by MLP. We use the \textbf{tanh} activation layer for hidden layers and the \textbf{softplus} activation function for the final layer of critics to ensure positiveness of the state value estimation.

\begin{table}[H]
\centering
\caption{Hyperparameters for the Adam Optimizer.}
\label{tab:adam_params}
\begin{tabular}{lll}
\toprule
\textbf{Parameter} & \textbf{Default Value} \\
\midrule
Exponential Decay ($\beta_1$) & $0.9$ \\
Exponential Decay ($\beta_2$) & $0.999$ \\
Epsilon ($\epsilon$) & $10^{-8}$\\
Weight Decay ($\lambda$) & $0$  \\
\bottomrule
\end{tabular}
\end{table}

\begin{table}[H]
\centering
\caption{Hyperparameters HATRPO's algorithms (\textbf{Harvest}).}
\label{tab:harvest_hatrpo}
\begin{tabular}{lllll}
\toprule
\textbf{Parameter} & \textbf{HATRPO} & \textbf{0.5-FHATRPO} & \textbf{1-FHATRPO} &\textbf{1.5-FHATRPO} \\
\midrule
$\nu$ & / & $0.1$ & $0.1$ & $0.1$ \\
KL ($\delta$)  & $0.01$ & $0.01$ & $0.01$ & $0.01$ \\
critic learning rate  & $0.0001$ & $0.0001$ & $0.0001$ & $0.0001$ \\
accept ratio & $0.1$ & $0.1$ & $0.1$ & $0.1$ \\
CG iteration & $15$ & $15$  & $15$  & $15$ \\
\bottomrule
\end{tabular}
\end{table}

\begin{table}[H]
\centering
\caption{Hyperparameters HAPPO's algorithms (\textbf{Harvest}).}
\label{tab:harvest_happo}
\begin{tabular}{lllll}
\toprule
\textbf{Parameter} & \textbf{HAPPO} & \textbf{0.5-FHAPPO} & \textbf{1-FHAPPO} &\textbf{1.5-FHAPPO} \\
\midrule
$\nu$ & / & $0.1$ & $0.1$ & $0.1$ \\
Clipping $\epsilon$& $0.05$ & $0.05$ & $0.05$ & $0.05$ \\
actor learning rate  & $0.0003$ & $0.0003$ & $0.0005$ & $0.0005$ \\
critic learning rate  & $0.0001$ & $0.0001$ & $0.0001$ & $0.0001$ \\
Minibatch size & $1000$ & $1000$ & $1000$ & $1000$ \\
\bottomrule
\end{tabular}
\end{table}

\begin{table}[H]
\centering
\caption{Hyperparameters HATRPO's algorithms (\textbf{CleanUp}).}
\label{tab:cleanup_hatrpo}
\begin{tabular}{lllll}
\toprule
\textbf{Parameter} & \textbf{HATRPO} & \textbf{0.5-FHATRPO} & \textbf{1-FHATRPO} &\textbf{1.5-FHATRPO} \\
\midrule
$\nu$ & / & $0.1$ & $0.1$ & $0.1$ \\
KL ($\delta$)  & $0.005$ & $0.005$ & $0.005$ & $0.0005$ \\
critic learning rate  & $0.0001$ & $0.0001$ & $0.0001$ & $0.0001$ \\
accept ratio & $0.1$ & $0.1$ & $0.1$ & $0.1$ \\
CG iteration & $10$ & $10$  & $10$  & $15$ \\
\bottomrule
\end{tabular}
\end{table}

\begin{table}[H]
\centering
\caption{Hyperparameters HAPPO's algorithms (\textbf{CleanUp}).}
\label{tab:cleanup_happo}
\begin{tabular}{lllll}
\toprule
\textbf{Parameter} & \textbf{HAPPO} & \textbf{0.5-FHAPPO} & \textbf{1-FHAPPO} &\textbf{1.5-FHAPPO} \\
\midrule
$\nu$ & / & $0.1$ & $0.1$ & $0.1$ \\
Clipping $\epsilon$& $0.1$ & $0.2$ & $0.1$ & $0.1$ \\
actor learning rate  & $0.0005$ & $0.0005$ & $0.0005$ & $0.0005$ \\
critic learning rate  & $0.0001$ & $0.0001$ & $0.0001$ & $0.0001$ \\
Minibatch size & $1000$ & $1000$ & $1000$ & $1000$ \\
\bottomrule
\end{tabular}
\end{table}

\begin{table}[h!]
\centering
\caption{Hyperparameters for FMAPPO algorithm.}
\label{tab:fmappo}
\begin{tabular}{lll}
\toprule
\textbf{Parameter} & \textbf{CleanUp} & \textbf{Harvest}  \\
\midrule
altruism level ($\alpha$) & $1$ & $1$ \\
Clipping $\epsilon$& $0.1$ & $0.1$ \\
actor learning rate  & $0.0005$ & $0.0005$ \\
critic learning rate  & $0.0005$ & $0.0005$ \\
Minibatch size & $1250$ & $1000$ \\
\bottomrule
\end{tabular}
\end{table}

\vspace{3cm}
\subsection{Additional metrics}
\label{app:add_metrics}

\textbf{CleanUp}

In the CleanUp environment, agents must clean the river and maintain pollution levels below a certain threshold to allow apples to spawn. Consequently, agents must sacrifice short-term rewards for the long-term benefit of cleaning the river. Furthermore, agents can utilize the zap action to teleport zapped agents to a random spawning location. We study this cooperation using two additional metrics: the Total Zap Action (TZA) and the Total Clean Action (TCA). Figure \ref{fig:add_cleanup} presents these additional results. The TZA tends to decrease throughout the training process across all algorithms, stabilizing between $50$ and $100$ actions per trajectory, while the TCA also decreases and stabilizes at approximately $200$ actions per trajectory.

\def\makezapcleanup#1/#2/#3{
    \addplot[name path=#1-upper, draw=none, forget plot] 
        table [col sep=comma, x=TRAINING TIMESTEP, y=#1_MAX] {cleanup_zap.csv};
        
    \addplot[name path=#1-lower, draw=none, forget plot] 
        table [col sep=comma, x=TRAINING TIMESTEP, y=#1_MIN] {cleanup_zap.csv};
        
    \addplot[fill=#2!20, opacity=0.6, forget plot] 
        fill between[of=#1-upper and #1-lower];
        
    \addplot[thick, #2, mark=none] 
        table [col sep=comma, x=TRAINING TIMESTEP, y=#1_MEAN] {cleanup_zap.csv};
        
    \addlegendentry{#3}
}

\def\makecleancleanup#1/#2/#3{
    \addplot[name path=#1-upper, draw=none, forget plot] 
        table [col sep=comma, x=TRAINING TIMESTEP, y=#1_MAX] {cleanup_clean.csv};
        
    \addplot[name path=#1-lower, draw=none, forget plot] 
        table [col sep=comma, x=TRAINING TIMESTEP, y=#1_MIN] {cleanup_clean.csv};
        
    \addplot[fill=#2!20, opacity=0.6, forget plot] 
        fill between[of=#1-upper and #1-lower];
        
    \addplot[thick, #2, mark=none] 
        table [col sep=comma, x=TRAINING TIMESTEP, y=#1_MEAN] {cleanup_clean.csv};
        
    \addlegendentry{#3}
}

\begin{figure}[h!]
    \centering
    
    \resizebox{\textwidth}{!}{
        
        \begin{tikzpicture}
            \begin{groupplot}[
                group style={
                    group size=3 by 2,
                    horizontal sep=1cm, 
                    vertical sep=1.7cm      
                },
                try min ticks=5,
                legend pos=north west,
                legend cell align=left,
                legend image code/.code={
                    \draw[mark repeat=2,mark phase=2, #1] 
                        plot coordinates {(0cm,0cm) (0.15cm,0cm) (0.3cm,0cm)};
                },
                title style={font=\scriptsize},         
                legend style={font=\fontsize{4}{4}\selectfont, row sep=-3pt, fill opacity=0.7,},   
                label style={font=\tiny}, 
                tick label style={font=\fontsize{4}{4}\selectfont},
                tick style={font=\fontsize{1}{4}\selectfont}, 
                width=6cm,  
                height=5cm,
                grid=major
            ]

                \nextgroupplot[title style={text width=3cm, align=center}, title={Total zap HATRPO/FHATRPO}, xlabel={timestep}, ylabel={Zap}]
                \pgfplotsinvokeforeach{HATRPO/red/{HATRPO}, FHATRPO_0.5/violet/{$\alpha$=0.5}, FHATRPO_1/cyan/{$\alpha$=1}, FHATRPO_1.5/magenta/{$\alpha$=1.5}}{
                    \makezapcleanup#1
                }
                
                \nextgroupplot[title style={text width=3cm, align=center}, title={Total zap consumed HAPPO/FHAPPO}, xlabel={timestep}, ylabel={Zap}]
                \pgfplotsinvokeforeach{HAPPO/blue/{HAPPO}, FHAPPO_0.5/orange/{$\alpha$=0.5}, FHAPPO_1/purple/{$\alpha$=1}, FHAPPO_1.5/teal/{$\alpha$=1.5}}{
                    \makezapcleanup#1
                }
                
                \nextgroupplot[title style={text width=3cm, align=center}, title={Total zap consumed for $\alpha=1$}, xlabel={timestep}, ylabel={Zap}]
                \pgfplotsinvokeforeach{FMAPPO/brown/{FMAPPO}, FHATRPO_1/cyan/{FHATRPO}, FHAPPO_1/purple/{FHAPPO}}{
                    \makezapcleanup#1
                }
                
                \nextgroupplot[title style={text width=3cm, align=center}, title={Total Clean actions HATRPO/FHATRPO}, xlabel={timestep}, ylabel={Clean}]
                \pgfplotsinvokeforeach{HATRPO/red/{HATRPO}, FHATRPO_0.5/violet/{$\alpha$=0.5}, FHATRPO_1/cyan/{$\alpha$=1}, FHATRPO_1.5/magenta/{$\alpha$=1.5}}{
                    \makecleancleanup#1
                }
                
                \nextgroupplot[title style={text width=3cm, align=center}, title={Total Clean actions HAPPO/FHAPPO}, xlabel={timestep}, ylabel={Clean}]
                \pgfplotsinvokeforeach{HAPPO/blue/{HAPPO}, FHAPPO_0.5/orange/{$\alpha$=0.5}, FHAPPO_1/purple/{$\alpha$=1}, FHAPPO_1.5/teal/{$\alpha$=1.5}}{
                    \makecleancleanup#1
                }
                
                \nextgroupplot[title style={text width=3cm, align=center}, title={Total Clean actions for $\alpha=1$}, xlabel={timestep}, ylabel={Clean}]
                \pgfplotsinvokeforeach{FMAPPO/brown/{FMAPPO}, FHATRPO_1/cyan/{FHATRPO}, FHAPPO_1/purple/{FHAPPO}}{
                    \makecleancleanup#1
                }
                
            \end{groupplot}
        \end{tikzpicture}
        
    } 
    
    \caption{Additional results on CleanUp. Each line is obtained by averaging its actual value on a rolling window of size 100 and its shaded area corresponds to its minimum and maximum on the same rolling window.}
    \label{fig:add_cleanup}
\end{figure}
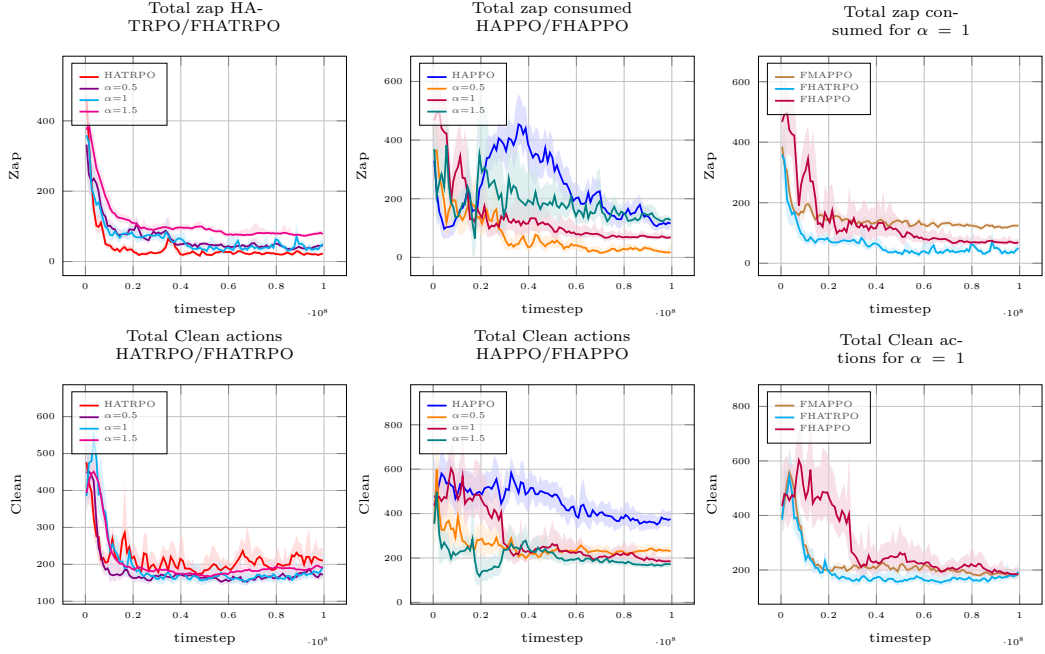

\textbf{Harvest}

In the Harvest environment, agents must harvest as many apples as possible while maintaining a sufficient apple population, as apples can only regenerate near existing ones. Similar to the CleanUp environment, agents can utilize the zap action to teleport zapped agents to a random spawning location. We can analyze this cooperation using two additional metrics: the Total Zap Action (TZA) and the Time to Depletion (TD). Figure \ref{fig:add_harvest} summarizes these additional results. The TZA values appear to decrease and approach zero for all algorithms except HAPPO and $0.5\text{-FHAPPO}$. Meanwhile, the TD exhibits an initial drop at the beginning of training—as agents learn to harvest apples efficiently—before increasing and plateauing at $500$, indicating that agents successfully maintain a stable level of apples on the board until the end of each trajectory.

\def\makezapharvest#1/#2/#3{
    \addplot[name path=#1-upper, draw=none, forget plot] 
        table [col sep=comma, x=TRAINING TIMESTEP, y=#1_MAX] {harvest_zap.csv};
        
    \addplot[name path=#1-lower, draw=none, forget plot] 
        table [col sep=comma, x=TRAINING TIMESTEP, y=#1_MIN] {harvest_zap.csv};
        
    \addplot[fill=#2!20, opacity=0.6, forget plot] 
        fill between[of=#1-upper and #1-lower];
        
    \addplot[thick, #2, mark=none] 
        table [col sep=comma, x=TRAINING TIMESTEP, y=#1_MEAN] {harvest_zap.csv};
        
    \addlegendentry{#3}
}

\def\makehitzero#1/#2/#3{
    \addplot[name path=#1-upper, draw=none, forget plot] 
        table [col sep=comma, x=TRAINING TIMESTEP, y=#1_MAX] {harvest_hit_zero.csv};
        
    \addplot[name path=#1-lower, draw=none, forget plot] 
        table [col sep=comma, x=TRAINING TIMESTEP, y=#1_MIN] {harvest_hit_zero.csv};
        
    \addplot[fill=#2!20, opacity=0.6, forget plot] 
        fill between[of=#1-upper and #1-lower];
        
    \addplot[thick, #2, mark=none] 
        table [col sep=comma, x=TRAINING TIMESTEP, y=#1_MEAN] {harvest_hit_zero.csv};
        
    \addlegendentry{#3}
}

\begin{figure}[h!]
    \centering
    
    \resizebox{\textwidth}{!}{
        
        \begin{tikzpicture}
            \begin{groupplot}[
                group style={
                    group size=3 by 2,
                    horizontal sep=1cm, 
                    vertical sep=1.7cm      
                },
                try min ticks=5,
                legend pos=north west,
                legend cell align=left,
                legend image code/.code={
                    \draw[mark repeat=2,mark phase=2, #1] 
                        plot coordinates {(0cm,0cm) (0.15cm,0cm) (0.3cm,0cm)};
                },
                title style={font=\scriptsize},         
                legend style={font=\fontsize{4}{4}\selectfont, row sep=-3pt, fill opacity=0.7,},   
                label style={font=\tiny}, 
                tick label style={font=\fontsize{4}{4}\selectfont},
                tick style={font=\fontsize{1}{4}\selectfont}, 
                width=6cm,  
                height=5cm,
                grid=major
            ]

                \nextgroupplot[title style={text width=3cm, align=center}, title={Total zap actions HATRPO/FHATRPO}, xlabel={timestep}, ylabel={Zap}]
                \pgfplotsinvokeforeach{HATRPO/red/{HATRPO}, FHATRPO_0.5/violet/{$\alpha$=0.5}, FHATRPO_1/cyan/{$\alpha$=1}, FHATRPO_1.5/magenta/{$\alpha$=1.5}}{
                    \makezapharvest#1
                }
                
                \nextgroupplot[title style={text width=3cm, align=center}, title={Total zap actions HAPPO/FHAPPO}, xlabel={timestep}, ylabel={Zap}]
                \pgfplotsinvokeforeach{HAPPO/blue/{HAPPO}, FHAPPO_0.5/orange/{$\alpha$=0.5}, FHAPPO_1/purple/{$\alpha$=1}, FHAPPO_1.5/teal/{$\alpha$=1.5}}{
                    \makezapharvest#1
                }
                
                \nextgroupplot[title style={text width=3cm, align=center}, title={Total zap actions for $\alpha=1$}, xlabel={timestep}, ylabel={Zap}]
                \pgfplotsinvokeforeach{FMAPPO/brown/{FMAPPO}, FHATRPO_1/cyan/{FHATRPO}, FHAPPO_1/purple/{FHAPPO}}{
                    \makezapharvest#1
                }
                
                \nextgroupplot[title style={text width=3cm, align=center}, title={Time to Depletion HATRPO/FHATRPO}, xlabel={timestep}, ylabel={Time to Depletion}]
                \pgfplotsinvokeforeach{HATRPO/red/{HATRPO}, FHATRPO_0.5/violet/{$\alpha$=0.5}, FHATRPO_1/cyan/{$\alpha$=1}, FHATRPO_1.5/magenta/{$\alpha$=1.5}}{
                    \makehitzero#1
                }
                
                \nextgroupplot[title style={text width=3cm, align=center}, title={Time to Depletion HAPPO/FHAPPO}, xlabel={timestep}, ylabel={Time to Depletion}]
                \pgfplotsinvokeforeach{HAPPO/blue/{HAPPO}, FHAPPO_0.5/orange/{$\alpha$=0.5}, FHAPPO_1/purple/{$\alpha$=1}, FHAPPO_1.5/teal/{$\alpha$=1.5}}{
                    \makehitzero#1
                }
                
                \nextgroupplot[title style={text width=3cm, align=center}, title={Time to Depletion for $\alpha=1$}, xlabel={timestep}, ylabel={Time to Depletion}]
                \pgfplotsinvokeforeach{FMAPPO/brown/{FMAPPO}, FHATRPO_1/cyan/{FHATRPO}, FHAPPO_1/purple/{FHAPPO}}{
                    \makehitzero#1
                }
                
            \end{groupplot}
        \end{tikzpicture}
        
    } 
    
    \caption{Additional results on Common Harvest. Each line is obtained by averaging its actual value on a rolling window of size 100 and its shaded area corresponds to its minimum and maximum on the same rolling window.}
    \label{fig:add_harvest}
\end{figure}
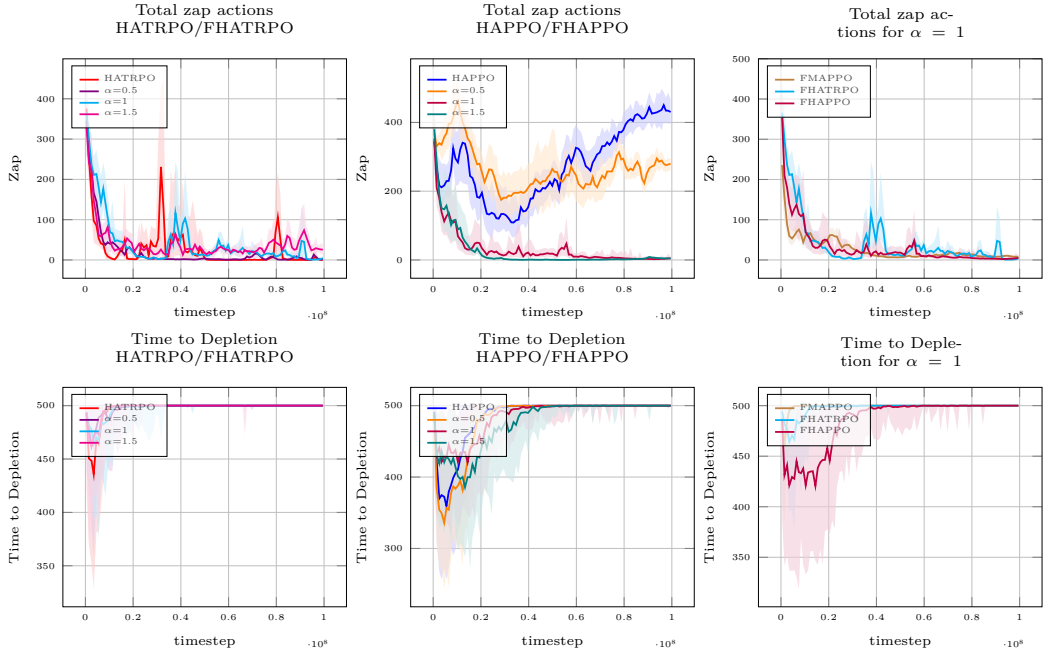


\end{document}